\documentclass{article}

\usepackage{amsfonts}
\usepackage{amssymb,epic,eepic}
\usepackage{amsmath}
\usepackage{dcolumn}
\usepackage{bm}
\usepackage{bbm}
\usepackage{verbatim}
\usepackage{color}
\usepackage{amsthm}
\usepackage{authblk}
\usepackage{tikz}
\usetikzlibrary{arrows,positioning,decorations.pathmorphing,decorations.markings,shapes,fadings}

\newcommand{\hr}{{\mathcal H}}

\newcommand{\cS}{{\mathcal S}}

\newcommand{\nn}{{\mathbb N}}

\newcommand{\bX}{\mathbf X}
\newcommand{\bY}{\mathbf Y}

\newcommand{\eins}{{\mathbbm{1}}}

\newtheorem{theorem}{Theorem}

\newtheorem{conjecture}{Conjecture}

\newtheorem{definition}{Definition}

\newtheorem{lemma}{Lemma}

\newtheorem{remark}{Remark}

\newcommand{\tr}{\mathrm{tr}}
\newcommand{\supp}{\mathrm{supp}}

\DeclareMathOperator{\spec}{spec}
\DeclareMathOperator{\pinch}{pinch}

\DeclareMathOperator{\sgn}{sgn}
\DeclareMathOperator{\linspan}{span}

\DeclareBoldMathCommand[heavy]\ka{(}
\DeclareBoldMathCommand[heavy]\kz{)}

\title{A class of permutation-invariant measurements and their relation to quantum relative entropies}
\author[1,2]{Janis N\"otzel}

\affil[1]{F\'{\i}sica Te\`{o}rica: Informaci\'{o} i Fen\`{o}mens Qu\`{a}ntics,\authorcr
Universitat Aut\`{o}noma de Barcelona,\authorcr
ES-08193 Bellaterra (Barcelona), Spain \authorcr
janis.notzel@uab.cat
\ \authorcr\ \authorcr\ }

\affil[2]{Lehrstuhl f\"{u}r Theoretische Informationstechnik,\authorcr
Technische Universit\"{a}t M\"{u}nchen,\authorcr 80333 Munich, Germany}

\begin{document}
\maketitle
\begin{abstract}
We characterize the asymptotic performance of a class of positive operator valued measurements (POVMs) where the only task is to make measurements on independent and identically distributed quantum states on finite-dimensional systems. The POVMs we utilize here can be efficiently described in terms of a reasonably small set of parameters. Their analysis furthers the development of a quantum method of types. They deliver provably optimal performance in asymmetric hypothesis testing and in the transmission of classical messages over quantum channels.\\
We now relate them to the recently developed $\alpha-z$ divergences $D_{\alpha,z}$ by giving an operational interpretation for the limiting case $\lim_{\alpha\to1}D_{\alpha,1-\alpha}$ in terms of probabilities for certain measurement outcomes. This explains one of the more surprising findings of \cite{audenaert-datta} in terms of the theory of group representations. In addition, we provide a Cauchy-Binet type formula for unitary matrices which connects the underlying representation theoretic objects to partial sums of the entries of unitary matrices.\\
At last, we concentrate on the special case of qubits. We are able to give a complete description of the asymptotic detection probabilities for all POVM elements described here. We take the opportunity to define a family of functions on pairs of semi-definite matrices which obeys the quantum generalizations of R\'enyi's axioms except from the generalized mean value axiom. This family is described by limiting values of $\alpha-z$ divergences for the extremal values of the parameter.
\end{abstract}
\begin{section}{Introduction}
The importance of representation theory for quantum information is understood best by taking a quick look at the structure of communication systems: Throughout, these systems employ certain \emph{structures} which are inserted into them by construction at sender's side and can then be detected by the receiver even if the signal gets corrupted by noise. The most general approach for making a signal received despite noise is to use some form of \emph{repetition}. Signals arising from repetition are obviously invariant under permutations. In a probabilistic sense, this property may continue to hold even under the influence of noise. Subsequently, the early and simple intuition of making communication resilient against noise by repetition has been developed to what is modern communication theory. The simple invariance under permutations that repetition delivers got lost in the hunt for higher throughput but the importance of group actions in code design prevailed.\\
Of course this argument is independent of the mathematical model we use for our signals and so an analogous reasoning applies to the study of signals which are described by quantum theory.\\
We concentrate on the study of finite dimensional quantum systems here. Measurements on such systems are completely described by positive operator valued measurements (POVMs). Such measurements have a finite number of different possible measurement outcomes. The probability that a POVM yields a specific outcome depends on the state that the system is in. In order to make communication tasks viable it is of great importance to deliver both sets $\{\rho_m\}_{m\in\mathbf M}$ of \emph{signals} which embody the messages $m\in\mathbf M$ into the quantum states $\rho_m$ and \emph{detection schemes} $\{D_m\}_{m\in\mathbf M}$ in the sense of POVMs such that the probability $\mathbf p_m$ of getting measurement outcome $m$ when the signal state is $m$ satisfies $\mathbf p_m\approx1$.\\
Besides these requirements it is necessary to deliver efficient descriptions of both the signals and the measurements.\\
We take this as an easy to grasp motivation to study the asymptotic performance of certain types of POVMs which are built up from certain well-described representations. This provides a way to obtain a clear specification of a class of quantum measurements in terms of a reasonably small (as compared to the frequency typical subsets of classical information theory. Our description needs less than twice the number of parameters needed for the description of a frequency typical subset with the underlying alphabet being $[d]$) set of parameters, an approach which opens up the possibility to deliver standardized detection procedures for quantum communication. Our main contribution in this direction is the characterization of the asymptotic detection probabilities of the POVMs studied here in terms of convex optimization problems which require optimization only on a finite number of copies of the underlying systems.\\
The POVM elements that we analyse here are capable of delivering provably optimal results both in asymmetric hypothesis testing \cite{noetzel-hypothesis} and in message transmission over noisy channels \cite{cbn}. They also deliver an intuitively appealing step forward in the development of a ``quantum method of types''. In this work we are further able to show initial connections to the geometry of positive matrices and to prove a surprising connection to the recently developed $\alpha-z$ divergences \cite{audenaert-datta}.\\
In the work \cite{noetzel-hypothesis} the author defined this specific class of POVMs as follows: For a given orthonormal basis $\{e_i\}_{i=1}^d$ of $\mathbb C^d$ and a ``frequency'' or ``type'' (a nonnegative function $f:\{1,\ldots,d\}\to\nn$ satisfying $\sum_if(i)=n$), consider the irreducible representations of the symmetric group $S_n$ on the frequency typical subspaces
\begin{align*}V_f:=\linspan(\{e_{i_1}\otimes\ldots\otimes e_{i_n}:|\{k:i_k=i\}|=f(i)\ \mathrm{for\ all\ }i\}.\end{align*}
Since these subspaces are invariant under permutations by definition, they naturally split up into different isotypical subrepresentations $V_{f,\lambda}$ (where $\lambda$ denote Young frames and some $V_{f,\lambda}$ may not contribute to above decomposition, meaning that $V_{f,\lambda}=\{0\}$ for some pairs $(f,\lambda)$, while some representations may occur several times)
\begin{align}V_f=\bigoplus_\lambda V_{f,\lambda}.\end{align}
The reason that the projections $P_{f,\lambda}$ onto these subspaces deliver optimal results in asymmetric hypothesis testing stems from the following observation: Given a state $\sigma$ with eigenvalues $t_1\geq t_2\geq\ldots\geq t_d>0$, we may pick one of its eigenbases for the definition of the $V_f$. It is then straightforward to show that the estimate
\begin{align}\tr\{P_{f,\lambda}\sigma^{\otimes n}\}&\approx\dim(F_\lambda)\cdot2^{-n\sum_i\tfrac{1}{n}f(i)\log t_i}\\
&\approx2^{n(H(\tfrac{1}{n}\lambda)-\sum_i\tfrac{1}{n}f(i)\log t_i)}
\end{align}
is valid, where $\dim(F_\lambda)$ is the dimension of the irreducible representation $F_\lambda$ of $S_n$ corresponding to $\lambda$ and $H(\tfrac{1}{n}\lambda)=-\sum_i\tfrac{1}{n}\lambda_i\log(\tfrac{1}{n}\lambda_i)$ is the entropy of the normalized Young frame. It turns out that, for $\lambda\approx n\cdot\mathrm{spec}(\rho)$ and $f(i)\approx n\cdot\langle e_i,\rho e_i\rangle\ \forall\ i\in[d]$ for some arbitrary second state $\rho$ we get
\begin{align}\tr\{P_{f,\lambda}\sigma^{\otimes n}\}&\approx2^{-nD(\rho\|\sigma)}.\end{align}
In our earlier work \cite{noetzel-hypothesis} we have not been able to characterize the exact asymptotic behaviour of the maps $(f,\lambda)\mapsto\tr\{P_{f,\lambda}\sigma^{\otimes n}\}$ when the basis $\{e_1,\ldots,e_d\}$ is no longer such that the representation of $\sigma$ in that basis is diagonal. We now start investigating this topic: For a fixed basis, two probability distributions $p$ and $s$ on $\{1,\ldots,d\}$ for which $s(1)\geq\ldots\geq s(d)$ we are interested in the map
\begin{align}
(\sigma,p,s)\mapsto-\lim_{n\to\infty}\frac{1}{n}\log\tr\{P_{f^{(n)},\lambda^{(n)}}\sigma^{\otimes n}\}
\end{align}
when $\tfrac{1}{n}f^{(n)}\to p$ and $\tfrac{1}{n}\lambda^{(n)}\to s$. The above limiting procedure is unfortunately not always easy to characterize. We do therefore not give a characterization in all detail and rather concentrate on the cases where $f^{(n)}=\lambda^{(n)}$ for all $n\in\nn$ or where $\sigma$ is diagonal in the chosen basis, whenever $d>2$. Only in the case $d=2$ we are able to exploit the specific nature of qubit systems to deliver a more detailed description.\\
Since every state $\sigma$ can be transformed into a state $U\sigma U^\dag$ which is diagonal in $\{e_1,\ldots,e_d\}$ we get a representation of $\sigma$ as
\begin{align}
\sigma^{\otimes n}=\sum_{f,\lambda}\frac{\tr\{(U \sigma U^\dag)^{\otimes n}P_f\}}{|T_f|}\cdot U^{\dag\otimes n} P_{f,\lambda} U^{\otimes n},
\end{align}
even if $\sigma$ is not diagonal in $\{e_1,\ldots,e_d\}$. Thus we felt motivated to also study the asymptotic behaviour of Hilbert-Schmidt scalar products of $P_{f,\lambda}$ with $A^{\otimes n}P_{f',\lambda}A^{\dag\otimes n}$ for arbitrary operators $A$. Especially in the case where $A=U$ is a unitary this led us to prove an interesting algebraic formula (see Lemma \ref{lemma:cauchy-binet-type-formula}) which arises from the study of the minimum of the function
\begin{align}
(U,p,q,s)\mapsto-\lim_{n\to\infty}\frac{1}{n}\log\tr\{P_{f^{(n)},\lambda^{(n)}}U^{\otimes n}P_{g^{(n)},\lambda^{(n)}}U^{\dag \otimes n}\}
\end{align}
where $\tfrac{1}{n}f^{(n)}\to p$, $\tfrac{1}{n}g^{(n)}\to q$ and $\tfrac{1}{n}\lambda^{(n)}\to s$. A detailed description of our approach is postponed to Section \ref{sec:main-results-and-definitions}, where we also give precise definitions of our main objects.\\
The method we utilize here opens up the possibility to split the analysis of the detection procedure into two parts: We note that $[P_f,P_\lambda]=0$ for all $f$ and $\lambda$. Thus, one can always realize e.g. the $P_f$ measurement first. This task is comparable to detection procedures in classical systems, if one takes local measurements (e.g. measurements where each POVM element is of the form $\otimes_{i=1}^nM_i$ for $M_1,\ldots,M_n\geq0$ and $M_1,\ldots,M_n\in M_d$ for free. It is not as clear how to implement the $P_\lambda$ measurements. A method for doing so has been outlined in \cite{harrow-thesis}, with its success being conditioned on the physical realizability of what is called the Schur-transform there (and in \cite{bacon-chuang-harrow}). Apart from \cite{harrow-thesis}, the work \cite{christandl-thesis} gives a lot of structural insights into the relations between representation theory and quantum information theory.\\
Recent work has also put into focus the definition of quantum relative entropies, and large families of such quantities have been defined and key properties like unitary invariance, convexity or monotonicity have been proven to hold.\\
We do not make any attempt to give a complete overview on the topic, we rather point the reader to the papers \cite{audenaert-datta,simon,hayashi-tomamichel}. These contain a good amount of the necessary history as well. Our focus here will be on the notion of $\alpha-z$ relative entropies $D_{\alpha,z}$ which are defined for real parameters $\alpha\neq1$ and $z\neq0$ in \cite{audenaert-datta}. It was proven by the authors of that work that their definition includes all the previous ones in the sense that certain choices of the parameters $\alpha$ and $z$ yield the other relative entropies.\\
We will pay special attention to the limits $\lim_{\alpha\to1}D_{\alpha,1-\alpha}$ and $\lim_{\alpha\to1}D_{\alpha,z}$ for $z\neq0$ which yield a rather cumbersome formula in the first case and the quantum relative entropy in the second. We are able here to give a direct operational interpretation of the first quantity in terms of an asymptotic probability (for details, see Section \ref{sec:main-results-and-definitions}) of obtaining certain measurement outcomes when the POVM which is utilized is of the form $\{P_{f,\lambda}\}_{f,\lambda}$ for some choice of basis. The second quantity is already connected via \cite{noetzel-hypothesis} and \cite{simon}. This motivates our definition of a huge set of non-negative functions $R_\cdot$ which fulfill some of the R\'enyi axioms (but not the generalized mean value axiom) and can be parameterized in such a way that they naturally include both $\lim_{\alpha\to1}D_{\alpha,1-\alpha}$ and $\lim_{\alpha\to1}D_{\alpha,z}$ for $z\neq0$. We point out a possible way to derive further connections to the geometry of positive matrices in Section \ref{sec:proof-that-axioms-are-fulfilled}.\\
A further study of this interplay between representation theory, information theory and the geometry of positive matrices is postponed to future work.
\end{section}
\begin{section}{Notation\label{sec:definitions}}
Throughout, $d\in\nn$ denotes the dimension of the vector space $\mathbb C^d$ that we will be working on. We use the convention $[d]:=\{1,\ldots,d\}$. For two natural numbers $d$ and $k$ satisfying $k\leq d$ and every $j\in[d]$ we set $[d]^k_j:=\{(i_1,\ldots,i_{k}):\exists!\ l\in[d]:\ i_l=j\}$. The set of probability distributions on $[d]$ is $\mathfrak P([d])$, and the set of ordered elements of $\mathfrak P([d])$ is $\mathfrak P^\downarrow([d]):=\{p\in\mathfrak P([d]):p(1)\geq\ldots\geq p(d)\}$. For an arbitrary $p\in\mathfrak P([d])$, $p^\downarrow\in\mathfrak P^\downarrow([d])$ is defined to have the same values as $p$, but in descending order. A function $g:[d]\to \mathbb R$ satisfies $g\succeq f$ for another function $f:[d]\to\mathbb R$ if $\sum_{i=1}^kg(i)\geq\sum_{i=1}^kf(i)$ for all $k=1,\ldots,d$.\\
The set of positive matrices acting on $\mathbb C^d$ is $\mathcal P_d$, the set of matrices is $M_d$. Non-square matrices are elements of $M_{d\times d'}$, where $d$ is the number of rows and $d'$ the number of columns. The adjoint of $b\in M_{d}$ written $b^\dag$.\\
$\cS(\mathbb C^d)$ is the set of states, i.e. positive semi-definite matrices with trace (the trace function on $M_d$ is denoted as $\tr$) $1$ acting on the Hilbert space $\mathbb C^d$. Pure states are given by projections onto one-dimensional subspaces. A vector $x\in\mathbb C^d$ of length one spanning such a subspace will therefore be referred to as a state vector, the corresponding state will be written
$|x\rangle\langle x|$. For a finite set $\mathbf X$, $|\mathbf X|$ denotes its cardinality. If $\mathbf X'\subset X$, then $\bX\backslash\bX':=\{x\in\bX:x\notin\bX'\}$. The extremal points of the convex set $\mathfrak P(\bX)$ are the distributions $\delta_x$ defined by $\delta_x(x')=1$ if and only if $x=x'$. We will also need ``channels'', meaning probability preserving linear maps from $\mathfrak P(\mathbf X)$ to $\mathfrak P(\bY)$. These are represented by matrices $W=(w(y|x))_{x\in\bX,y\in\bY}$ which satisfy $\sum_{y\in\bY}w(y|x)=1$ for all $x\in\bX$. Their action is uniquely defined by setting $W(\delta_x):=\sum_{y\in\bY}w(y|x)\delta_y$ for every $x\in\bX$. The set of channels from $\bX$ to $\bY$ is denoted $C(\mathbf X,\mathbf Y)$.<\\
For any $n\in\nn$, we define $\bX^n:=\{(x_1,\ldots,x_n):x_i\in\bX\ \forall i\in\{1,\ldots,n\}\}$, we also write $x^n$ for the elements of $\bX^n$. Given such element, $N(\cdot|x^n)$ denotes its type, and is defined through $\forall x\in\bX:\ N(x|x^n):=|\{i:x_i=x\}|$. The set of all types arising from words of length $n$ is written $\mathbbm T_n$ or, if the alphabet is not clear from the context, $\mathbbm T_n(\bX)$.\\
The von Neumann entropy of a state $\rho\in\mathcal{S}(\hr)$ is given by
\begin{equation}S(\rho):=-\textrm{tr}(\rho \log\rho),\end{equation}
where $\log(\cdot)$ denotes the base two logarithm which is used throughout the paper. The  entropy of $r\in\mathfrak P(\bX)$ is defined by the formula
\begin{align}
H(r):=-\sum_{x\in \bX}r(x)\log(r(x)).
\end{align}
Given two states $\rho,\sigma\in\cS(\mathbb C^d)$, the relative entropy of them is defined as
\begin{align}\label{eqn:def-of-relative-entropy}
D(\rho\|\sigma):=\left\{\begin{array}{l l}\tr\{\rho(\log(\rho)-\log(\sigma))\},&\mathrm{if}\ \supp(\rho)\subset\supp(\sigma),\\\infty,&\mathrm{else} \end{array}\right.
\end{align}
For $p,q\in\mathfrak P([d])$ we may set $\rho:=\sum_{i=1}^dp(i)|e_i\rangle\langle e_i|$ and $\sigma:=\sum_{i=1}^dq(i)|e_i\rangle\langle e_i|$, then (with a slight abuse of notation ) $D(p\|q):=D(\rho\|\sigma)$ defines the classical Kullback-Leibler distance $D(p\|q)$ between probability distributions $p,q\in\mathfrak P([d])$ as well.\\
We now fix our notation for representation theoretic objects and state some basic facts.\\
The symbols $\lambda,\mu$ will be used to denote Young frames. The set of Young frames with at most $d\in\nn$ rows and $n\in\nn$ boxes is denoted $\mathbb Y_{d,n}$.\\
For any given $n$, the representation of $S_n$ we will consider is the standard representation on $(\mathbb C^d)^{\otimes n}$ that acts by permuting tensor factors.\\
The most important technical definition for this work is that of frequency-typical subspaces $V_f$ of $(\mathbb C^d)^{\otimes n}$. These arise from choosing a fixed orthonormal basis $\{e_i\}_{i=1}^d$ of $\mathbb C^d$, choosing a frequency $f$ (a function $f:[d]\to\mathbb N$ satisfying $\sum_{i=1}^df(i)=n$), setting $T_f:=\{(i_1,\ldots,i_n):|\{i_k:i_k=j\}|=f(j)\ \forall j\in[d]\}$, and defining
\begin{align}
V_f:=\linspan(\{e_{i_1}\otimes\ldots\otimes e_{i_n}:(i_1,\ldots,i_n)\in T_f\}).
\end{align}
They have been widely used in quantum information theory, but share one very nice property that has not been explicitly exploited in quantum information theory until \cite{noetzel-hypothesis}: They are invariant under permutations, if the (linear) action $\mathbb B$ of $S_n$ on $(\mathbb C^d)^{\otimes n}$ is defined in the natural way via
\begin{align}
\mathbb B(\tau)v_1\otimes\ldots\otimes v_n:=v_{\tau^{-1}(1)}\otimes\ldots\otimes v_{\tau^{-1}(n)}
\end{align}
for all $\tau\in S_n$ and $v_1,\ldots,v_n\in\mathbb C^d$. From the invariance of each $V_f$ under the action $\mathbb B$ of $S_n$ it immediately follows that
\begin{align}
V_f=\bigoplus_\lambda V_{f,\lambda},
\end{align}
where each $V_{f,\lambda}$ is just a direct sum of irreducible representations corresponding to $\lambda$ that is contained entirely within $V_f$. The multiplicity of $F_\lambda$ within $V_f$ is given by $\dim(V_{f,\lambda})/\dim(F_\lambda)$. It is a number which scales at most polynomially in $n$, if $d$ is kept fixed. The quantity $F_\lambda$ denotes the unique complex vector space carrying the irreducible representation of $S_n$ corresponding to a Young frame $\lambda$. Each such $\lambda$ consists of $n\in\nn$ boxes and has row lengths $\lambda_1,\ldots,\lambda_d$ for some $d\in\nn$. Thus, $\bar\lambda$ defined by $\bar\lambda(i):=\tfrac{1}{n}\lambda_i$ for every $i\in[d]$ defines an element of $\mathfrak P^\downarrow([d])$.
\\\\
During our analysis it turns out that, for every $k\leq d$, the vectors
\begin{align}
v_k:=\frac{1}{\sqrt{k}}\sum_{\tau\in S_k}\sgn(\tau)\mathbb B(\tau)e_1\otimes\ldots\otimes e_k
\end{align}
are important. Also, we are going to employ the following estimate taken from \cite[Lemma 2.3]{csiszar-koerner}, which is valid for all frequencies $f:[d]\to\nn$ that satisfy $\sum_{i=1}^df(i)=n$:
\begin{equation}
\frac{1}{(n+1)^d}2^{nH(\overline{f})}\leq|T_f|\leq2^{nH(\overline{f})}\label{eqn2},
\end{equation}
where $\bar f:=\tfrac{1}{n}f$. We will also need \cite[Lemma 2.7]{csiszar-koerner} which employs the variational distance that we define as $\|p-q\|:=\sum_{x\in\bX}|p(x)-q(x)|$ for all $p,q\in\mathfrak P(\bX)$ and delivers:
\begin{lemma}\label{lemma1}
If, for $\mathbf X$ a finite alphabet and $p,q\in\mathfrak P(\mathbf X)$ we have $\|p-q\|\leq\Theta\leq1/2$, then
\begin{equation}
 |H(p)-H(q)|\leq-\Theta\log\frac{\Theta}{|\mathbf X|}.
\end{equation}
\end{lemma}
Another very important estimate is the following one (a derivation can e.g. be found in \cite{noetzel-2ptypicality}):
\begin{equation}
 2^{n(H(\bar\lambda)-\frac{2d^6}{n}\log(2n))}\leq\dim F_\lambda\leq 2^{nH(\bar\lambda)}\qquad (\lambda\in\mathbb Y_{d,n}).\label{eqn4}
\end{equation}
During our investigation we shall need the following sets of distributions: For every $q\in\mathfrak P([d])$ and $1\leq k\leq d$, set
\begin{align}
\mathfrak P(q,k):=\left\{p\in\mathfrak P([d]^k):\begin{array}{l}N(i|(i_1,\ldots,i_k))>1\ \Rightarrow p((i_1,\ldots,i_k))=0\\ p([d]^k_i)=k\cdot q(i)\ \forall\ i\in[d]\end{array}\right\}.
\end{align}
Such distributions can be constructed by taking a unitary $U\in M_d$ and defining $p\in\mathfrak P([d]^k)$ via $p((i_1,\ldots,i_k)):=|\langle e_{i_1}\otimes\ldots e_{i_k},U^{\otimes k}v_k\rangle|^2$. This ensures the validity of $N(i|(i_1,\ldots,i_k))>1\ \Rightarrow p((i_1,\ldots,i_k))=0$. Lemma \ref{lemma:cauchy-binet-type-formula} then delivers the values of $p([d]^k_i)$. This connection demonstrates that $\mathfrak P(q,k)\neq\emptyset$ is possible, thus making our definition nontrivial.\\
We now switch the topic one last time in this section and concentrate on additional entropic quantities which are necessary in the remainder:
\begin{definition}[Reverse Sandwiched Relative Entropy]\label{def:reverse-sandwiched-renyi-entropy}
For $\rho,\sigma\in\cS(\mathbb C^d)$ with $\supp\rho\subset\supp\sigma$ and $\alpha\in\mathbb R\backslash\{0\}$, set
\begin{align}
\hat D_\alpha(\rho\|\sigma):=\frac{1}{\alpha-1}\log\tr\left\{\left(\rho^{\frac{\alpha}{2(1-\alpha)}}\sigma\rho^{\frac{\alpha}{2(1-\alpha)}}\right)^{1-\alpha}\right\}.
\end{align}
\end{definition}
The revere sandwiched relative entropy is derived from the sandwiched relative entropy, which was defined in \cite{wilde-winter-yang} and \cite{lennert-dupuis-szehr-fehr-tomamichel}:
\begin{definition}[Sandwiched Relative Entropy]\label{def:sandwiched-renyi-entropy}
For $\rho,\sigma\in\cS(\mathbb C^d)$ with $\supp\rho\subset\supp\sigma$ and $\alpha\in\mathbb R\backslash\{0\}$, set
\begin{align}
\tilde D_\alpha(\rho\|\sigma):=\frac{1}{\alpha-1}\log\tr\left\{\left(\sigma^{\frac{1-\alpha}{2\alpha}}\rho\sigma^{\frac{1-\alpha}{2\alpha}}\right)^{\alpha}\right\}.
\end{align}
\end{definition}
As was made explicit in \cite[Equation 11]{audenaert-datta}, the two quantities are related through the following equation:
\begin{align}
(\alpha-1)\hat D_\alpha(\rho\|\sigma)=(-\alpha)\tilde D_{1-\alpha}(\sigma\|\rho).
\end{align}
The sandwiched relative entropy $\tilde D_\alpha$ has been proven to have a huge number of highly desirable properties in, among others, the work \cite{beigi} and \cite{lennert-dupuis-szehr-fehr-tomamichel}. From one of its origins, it is intimately connected to quantum channel coding \cite{wilde-winter-yang}. Applications are also found in hypothesis testing \cite{mosonyi}, \cite{mosonyi-ogawa}, \cite{mosonyi-ogawa-2}, \cite{hayashi-tomamichel}.\\
A more general definition was made by Audenaert and Datta \cite{audenaert-datta}. It adds the parameter $z>0$ and reads
\begin{definition}[$\alpha-z$ relative entropy]
For $\rho,\sigma\in\cS(\mathbb C^d)$ with $\supp\rho\subset\supp\sigma$, $\alpha\in\mathbb R\backslash\{0\}$ and $z>0$, set
\begin{align}
D_{\alpha,z}(\rho\|\sigma):=\frac{1}{\alpha-1}\log\tr\left\{\left(\rho^{\alpha/z}\sigma^{(1-\alpha)/z}\right)^z\right\}.
\end{align}
\end{definition}
The work \cite{audenaert-datta} not only defined this quantity but also provided a lot of details and especially gave an explicit formula for the limit $\lim_{\alpha\to1}D_{\alpha,1-\alpha}(\rho\|\sigma)$ that we shall use in the proof of Theorem \ref{main-theorem}.
\end{section}
\begin{section}{Main Results and Definitions\label{sec:main-results-and-definitions}}
Throughout, we make our definitions with respect to one fixed but arbitrary orthonormal basis $B_d=\{e_1,\ldots,e_d\}$ - the standard basis within $\mathbb C^d$. Every matrix and also every quantum state are represented with respect to that basis.
\begin{definition}\label{def:corner-point-AudDat}
To any pair $(\rho,\sigma)\in\mathcal S(\mathbb C^d)\times\mathcal S(\mathbb C^d)$ we assign a unitary transformation $U_\rho$ such that $U_\rho\rho U_\rho^\dag$ is diagonal and has its diagonal entries sorted in descending order. We can then define the function
\begin{align}
\Phi&\ :\ \cS(\mathbb C^d)\times\mathbb \cS(C^d)\to\mathbb R_+,\\
(\rho,\sigma)&\mapsto-\lim_{n\to\infty}\frac{1}{n}\log\tr\{U_\rho^{\dag\otimes n}P_{f^{(n)},\lambda^{(n)}}U_\rho^{\otimes n}\sigma^{\otimes n}\},
\end{align}
where the sequence $(f^{(n)},\lambda^{(n)})_{n\in\nn}$ satisfies $\lim_{n\to\infty}\tfrac{1}{n}\lambda^{(n)}=(\langle e_1,\rho e_1\rangle,\ldots,\langle e_d,\rho e_d\rangle$ and $f^{(n)}=\lambda^{(n)}$ for all $n\in\mathbb N$.
\end{definition}
Although this definition is ambiguous whenever $\rho$ has degenerate eigenvalues we show later that it is still well-defined.
\begin{definition}\label{def:corner-point-SteLem}
In the same way as in Definition \ref{def:corner-point-AudDat} we take any unitary matrix $U_\sigma$ such that $U_\sigma^\dag\sigma U_\sigma$ is diagonal in $B_d$. We define the function $\Lambda:\mathcal S(\mathbb C^d)\times\mathcal S(\mathbb C^d)\to\mathbb R_+$ by taking any sequence of frequencies satisfying $\lim_{n\to\infty}\tfrac{1}{n}f^{(n)}=(\langle e_1,U_\sigma\rho U_\sigma^\dag e_1\rangle,\ldots,\langle e_d,U_\sigma\rho U_\sigma^\dag e_d\rangle)$ and a sequence of Young frames satisfying $\lim_{n\to\infty}\tfrac{1}{n}\lambda^{(n)}=\spec(\rho)$. We then set
\begin{align}
\Lambda&\ :\ \cS(\mathbb C^d)\times\cS(\mathbb C^d)\to\mathbb R_+,\\
(\rho,\sigma)&\mapsto-\lim_{n\to\infty}\frac{1}{n}\log\tr\{U_\sigma^{\otimes n}P_{f^{(n)},\lambda^{(n)}}U_\sigma^{\dag\otimes n}\sigma^{\otimes n}\}.
\end{align}
Again, we show later that this definition does not depend on a particular one among the many possible choices of $U_\sigma$.
\end{definition}
The third quantity we define is
\begin{definition}\label{def:Delta_A} Let $d\in\nn$ and $A\in M_d$. Define a function $\Delta_A$ as
\begin{align}
\Delta_A:\mathfrak P([d])\times\mathfrak P^\downarrow([d])\times\cS(\mathbb C^d)\to\mathbb R_+,
\end{align}
\begin{align}
(p,s,\sigma)\mapsto\left\{\begin{array}{ll}-\lim_{n\to\infty}\frac{1}{n}\log\tr\{A^{\otimes n}P_{f^{(n)},\lambda^{(n)}}A^{\dag\otimes n}\sigma^{\otimes n}\},&\ \mathrm{if}\ s\succeq p^\downarrow\\
\infty,&\ \mathrm{else}\end{array}\right.
\end{align}
where $(f^{(n)})_{n\in\nn}$ and $(\lambda^{(n)})_{n\in\nn}$ are sequences satisfying $\lim_{n\to\infty}\tfrac{1}{n}\lambda^{(n)}=s\in\mathfrak P^\downarrow([d])$, $\lim_{n\to\infty}\tfrac{1}{n}f^{(n)}=p\in\mathfrak P([d])$ and the sequences are constructed such that $\lambda^{(n)}\succeq f^{(n)\downarrow}$ for all $n\in\nn$.
\end{definition}
Of course $\Delta_{U}=\Phi$ whenever $U$ is a unitary matrix and $U\rho U^\dag$ is diagonal in $B_d$ and has decreasing diagonal entries and $s=\spec\rho$. Also, $\Delta_U=\Lambda$ whenever $U$ is unitary and $U^\dag\sigma U$ is diagonal in $B$ and $p=\pinch\rho$, $s=\spec\rho$. In the remaining cases it is not clear from their definition that $\Phi$, $\Lambda$ or $\Delta$ are well defined. Note that $s\succeq p^\downarrow$ implies the existence of sequences $(\lambda^{(n)})_{n\in\nn}$ and $f^{(n)})_{n\in\nn}$ with respective limits $s$ and $p$ and such that the Kostka numbers $K_{\lambda^{(n)},f^{(n)\downarrow}}$ of these sequences are non-negative \cite[Exercise 2]{fulton}). If $K_{\lambda^{(n)},f^{(n)}}>0$ however then the construction provided in \cite[Chapter 5.5]{sternberg} proves that $V_{f^{(n)},\lambda^{(n)}}\neq\{0\}$. We will use this connection more explicitly in the proofs of Theorems \ref{main-theorem} and \ref{theorem:characterization-of-DU}.\\
The question whether $\Delta_A$ is well-defined in the sense of being independent from the specific sequences $(f^{(n)})_{n\in\nn}$ and $(\lambda^{(n)})_{n\in\nn}$ will be settled here only for $d=2$ in Theorem \ref{theorem:characterization-of-DU}. We are thus left with a conjecture:
\begin{conjecture}
The functions $\Delta_A$ from Definition \ref{def:Delta_A} are well-defined for every $d\geq1$ and every $A\in M_d$.
\end{conjecture}
For the other two quantities it will become immediate that they are well-defined once we calculate the limits in the Definitions \ref{def:corner-point-AudDat} and \ref{def:corner-point-SteLem}. This task leads us to the following theorem:
 \begin{theorem}\label{main-theorem}
For every two states $\rho,\sigma\in\cS(\mathbb C^d)$ we have
\begin{enumerate}
\item
$\begin{aligned}
\Phi(\rho\|\sigma)=\lim_{\alpha\to1}\hat D_{\alpha}(\rho\|\sigma),
\end{aligned}$
\item
$\begin{aligned}
\Lambda(\rho\|\sigma)=D(\rho\|\sigma).
\end{aligned}$
\end{enumerate}
\end{theorem}
\begin{remark}
The second of the above statements has been proven in \cite{noetzel-hypothesis} and will not be proven here again.
\end{remark}
This result raises some interest into a more in-depth study of the projections $P_{f,\lambda}$. As we already observed before \cite{noetzel-hypothesis} the subspaces $V_{f,\lambda}$ are generically not irreducible. In these cases computations are less straightforward as in the cases where we have irreducible representations. We therefore concentrate here on a study of cases where at least one of the subspaces involved into the calculation is irreducible. The decomposition
\begin{align}
\sigma^{\otimes n}=\sum_{f,\lambda}\frac{\tr\{P_{f}(U\sigma U^\dag)^{\otimes n}\}}{|T_f|}\cdot U^{\dag\otimes n} P_{f,\lambda}U^{\otimes n}
\end{align}
which is valid whenever $\sigma$ is diagonal in $B_d$ motivates the study of objects of the form $\tr\{P_{f,\lambda}U^{\otimes n}P_{f',\lambda'}U^{\dag\otimes n}\}=\tr\{P_{f,\lambda}U^{\otimes n}P_{f',\lambda}U^{\dag\otimes n}\}\delta(\lambda,\lambda')$ where $U\in M_d$ is unitary. Of course this is done again in the asymptotic setting, and with a slight increase in generality:
\begin{definition}\label{def:scalar-products-for-arbitrary-d}
For every $d\geq2$, asymptotic shapes $s$ of Young frame, frequencies $q\in\mathfrak P([d])$, and matrix $A\in M_{d}$ we define
\begin{align}
\Theta(q,s,A):=-\lim_{n\to\infty}\frac{1}{n}\log\tr\{P_{\lambda^{(n)},\lambda^{(n)}}A^{\otimes n}P_{f^{(n)},\lambda^{(n)}}A^{\dag\otimes n}\},
\end{align}
where $(\lambda^{(n)})_{n\in\nn}$ is a sequence of Young frames satisfying $\lim_{n\to\infty}\tfrac{1}{n}\lambda^{(n)}=s$ and $(f^{(n)})_{n\in\nn}$ a sequence of frequencies satisfying $\lim_{n\to\infty}\tfrac{1}{n}f^{(n)}=q$.
\end{definition}
The results we obtain from the study of $\Theta$ are presented in the next theorem:
\begin{theorem}\label{theorem:scalar-products-for-arbitrary-d}
For every $d\geq2$ and asymptotic shape $s$ of Young frames as well as asymptotic frequency $q$ and $A\in M_d$ the function $\Theta$ assumes the value
\begin{align}
\Theta(q,s,A)=H(s)-\min_{W(\hat s)=q}\sum_{k=1}^d\hat s(k)\min_{p\in\mathfrak P(W(\delta_k),k)}\tfrac{1}{k}D(p\|p_k)
\end{align}
where $p_k(i_1,\ldots,i_k):=|\langle e_{i_1}\otimes\ldots\otimes e_{i_k},A^{\otimes k}\frac{1}{\sqrt{k}}\sum_{\tau\in S_k}\sgn(\tau)\mathbb B(\tau)\left(e_1\otimes\ldots e_k\right)\rangle|^2$ and $\hat s\in\mathfrak P([d])$ is defined by $\hat s(k):=\left(s(k)-s(k+1)\right)\cdot k$ for all $k\in[d]$ using the convention $s(d+1):=0$. For every unitary matrix $U$ and fixed asymptotic shape $s$, the function $q\mapsto\Theta(q,s,U)$ assumes its minimum at a distribution $\tilde q$ which satisfies $\tilde q(i):=\sum_{k=1}^d\hat s(k)\sum_{l=1}^k\tfrac{|u_{il}|^2}{k}$.
\end{theorem}
\begin{remark}Note that the map $k\mapsto \sum_{l=1}^k\frac{|u_{il}|^2}{k}$ actually defines an element of $C([d],[d])$ since for every $k$ we have $\sum_{i=1}^d\sum_{l=1}^k\frac{|u_{il}|^2}{k}=1$.\end{remark}
Especially the location of the minimum which we describe above made us conjecture an interesting formula via the following route: An application of Pinsker's inequality to above formula for $\Theta$ lets us transform the search for the minimum into a question about distance in norm rather than relative entropy. This decomposition delivers a lower bound which can be shown to equal zero if and only if $q$ has the desired form:
\begin{lemma}[Estimate for norms\label{lemma:norm-estimate}] Let $s\in\mathfrak P^\downarrow([d])$ and define $\hat s\in\mathfrak P([d])$ by $\hat s(k)=\left(s(k)-s(k+1)\right)\cdot k$ for all $k\in[d]$ and using the convention $s(d+1):=0$. Let there be distributions $q,q_1,\ldots,q_k\in\mathfrak P([d])$ such that $\sum_{k=1}^d\hat s(k)q_k=q$. Let further $U\in M_{d\times d}$ be a unitary matrix and $\tilde q\in\mathfrak P([d])$ be defined by $\tilde q(i):=\sum_{k=1}^d\hat s(k)\sum_{l=1}^k\tfrac{|u_{il}|^2}{k}$. For every $k\in[d]$ assume that there exists a $p(\cdot|k)\in\mathfrak P(q_k,k)$. Let finally $p_k\in \mathfrak P([d]^k)$ be defined by $p_k(i_1,\ldots,i_k):=|\langle e_{i_1}\otimes\ldots\otimes e_{i_k},U^{\otimes k}v_k\rangle|^2$. It holds
\begin{align}
\sum_{k=1}^d\hat s(k)\cdot\|p(\cdot|k)-p_k\|&\geq\|q-\tilde q\|.
\end{align}
\end{lemma}
\begin{remark}
The lemma gains its proper interpretation by letting $\langle e_i,\rho e_i\rangle=\sum_{k=1}^ds(k)|u_{ik}|^2$ be the pinching of some state $\rho$ with spectrum $s$ to the chosen basis.
\end{remark}
The proof of Lemma \ref{lemma:norm-estimate} rests on the validity of the following version of the Cauchy-Binet formula which seems to have a certain worth in its own right:
\begin{lemma}\label{lemma:cauchy-binet-type-formula}
Let $U\in M_{d}$ be unitary. For every $k\leq d$ and $j\in[d]$ it holds
\begin{align}
\sum_{(i_1,\ldots,i_{k})\in[d]^k_j}|\langle e_{i_1}\otimes\ldots\otimes e_{i_k},U^{\otimes k}v_{k}\rangle|^2=\sum_{i=1}^{k}|u_{ji}|^2
\end{align}
where, as defined in Section \ref{sec:definitions}, $v_{k}:=\frac{1}{\sqrt{k}}\sum_{\tau\in S_k}\sgn(\tau)\mathbb B(\tau)e_1\otimes\ldots\otimes e_{k}$.
\end{lemma}
Some of our results can only be proven to hold for $d=2$, where every representation $V_{f,\lambda}$ is irreducible. These results are summarized below. We start with the definition of two basic building blocks of our analysis:
\begin{definition}\label{def:building-blocks-for-d=2}
Let $A\in\mathcal B(\mathbb C^2)$ and $q,p\in\mathfrak P([2])$. Then we set
\begin{align}
\Theta_1(A,q,p):=-\frac{1}{n}\log\langle v_{f^{(n)}},A^{\otimes n}P_{g^{(n)}}A^{\dag\otimes n}v_{f^{(n)}}\rangle
\end{align}
where $\lim_{n\to\infty}\tfrac{1}{n}f^{(n)}=p$ and $\lim_{n\to\infty}\tfrac{1}{n}g^{(n)}=q$. We also set
\begin{align}
\Theta_2(A,p,q):=-\lim_{n\to\infty}\frac{1}{2\cdot n}\log\langle v_2^{\otimes n},A^{\otimes 2\cdot n},P_{g^{(2\cdot n)}}A^{\otimes 2\cdot n}v_2^{\otimes n}\rangle,
\end{align}
where by definition of the vectors $v_k$ we have $v_2=\sqrt{1/2}(e_1\otimes e_2-e_2\otimes e_1)$, and finally
\begin{align}
\Theta(p,q,s,A):=-\lim_{n\to\infty}\frac{1}{n}\log\tr\{P_{f^{(n)},\lambda^{(n)}}A^{\otimes n}P_{g^{(n)},\lambda^{(n)}}A^{\dag\otimes n}\}
\end{align}
where $\lim_{n\to\infty}\tfrac{1}{n}\lambda^{(n)}=s\in\mathfrak P^\downarrow([2])$.
\end{definition}
Note that $\Theta_2$ does not really depend on $p$ and that the latter parameter is only included into the definition in order to be able to deliver a complete description of all quantities within a unified setting.\\
The quantities $\Theta$, $\Theta_1$ and $\Theta_2$ are connected via the following theorem:
\begin{theorem}\label{theorem:scalar-products-for-d=2}
Let $\lim_{n\to\infty}\tfrac{1}{n}f^{(n)}=p$ and $\lim_{n\to\infty}\tfrac{1}{n}g^{(n)}=q$ as well as $\lim_{n\to\infty}\tfrac{1}{n}\lambda^{(n)}=s$ and $A\in M_{d}$. Then
\begin{align}
\Theta_1(A,p,q)=H(q)+\min_{r\in\Xi(p,q)\}}D(r\|p_{1,A}),\qquad\Theta_2(A,p,q)=-\tfrac{1}{2}\min_{r\in\mathfrak P(2,q)}D(r\|p_{2,A})
\end{align}
where $\Xi(p,q):=\{r\in\mathfrak P([2]\times[2]):r_1=p,\ r_2=q\}$ and $p_{1,A}$ is defined by $p_{1,A}(i,j):=|\langle e_i,A^\dag e_j\rangle|^2$ and $p_{2,A}(i,j):=|\langle e_i\otimes e_j,A^\dag v_2\rangle|^2$. In addition to that,
\begin{align}
\Theta(p,q,s,A)=\min_{W(p)=q}\sum_{i=1}^2\hat s(i)\cdot\Theta_i(A,p,W(\delta_i)).
\end{align}
\end{theorem}
In the case $d=2$ we are able to give a characterization of $\Delta_U$:
\begin{theorem}\label{theorem:characterization-of-DU}
Let $d=2$ and $A\in M_2$. The function $\Delta_A$ satisfies, for all $p,s\in\mathfrak P^\downarrow([d])$ and $\sigma\in\mathcal B(\mathbb C^2)$, the following:
\begin{align}
\Delta_A(p,s,\sigma)=-H(s)-\hat s(2)\cdot\log\det A\sigma A^\dag+\hat s(1)\cdot\bar D((\tfrac{p(1)-p(2)}{\hat s(2)},\tfrac{p(2)-s(2)}{\hat s(2)})\|A\sigma A^\dag),
\end{align}
where the function $\bar D:\mathfrak P([2])\times M_2\to\mathbb R_+$ is given by the convex optimization problem
\begin{align}
\bar D(p,X):=\min_{W:W(p)=p}\sum_{j=1,2}p(j)D(W(\delta_j)\||X_{\cdot j}|).
\end{align}
\end{theorem}
\begin{remark}
Of course this formula demonstrates that $\Delta_\cdot$ is continuous in $A$, $p$ and $s$ - on the region of parameters satisfying $s(2)< p(2)$. The usefulness of the formula also stems from the fact that it allows an explicit and efficient computation of the probability that the state $\sigma$ is detected by a measurement scheme which asymptotically detects states with pinching $p$ and spectrum $s$.\\
While it is clear that the function $\bar D$ delivers an efficient way of computing the values of the function $\Delta_U$ for all unitaries (in fact, for all $A\in M_{d}$), we have not been able to deliver more insightful reformulations of it. We note that interesting connections to matrix scaling (see e.g. the recent work \cite{kurras}) are given, and another interesting connection is that to information projections in the sense of \cite{csiszar-I-projections}.\\
It seems tempting to look for connections to $\lim_{\alpha\to1}D_{\alpha,z}$ for $z\in(0,1]$ but it has already been proven that these limits are all equal to $D$ in \cite{simon}. Another possible route would be to look at the limits $\lim_{\alpha\to1}D_{\alpha,r(1-\alpha)}$ which we were able to use for $r=1$, but we have not been able so far to find any relations of these quantities to $\Delta_A$ so far.
\end{remark}
At last we exhaust the peculiarities of the case $d=2$ to define a set of functions which do, to some extent, measure the distance between two states $\rho$ and $\sigma$. From the very start, they have offer an operational interpretation. In special limiting cases, they deliver either $D(\rho\|\sigma)$ or $\lim_{\alpha\to1}\hat D_{\alpha}(\rho\|\sigma)$.\\
In section \ref{sec:proof-that-axioms-are-fulfilled} we provide a proof that they fulfill all the R\'enyi axioms except the generalized mean value axiom, which is no surprise given that they are well-defined also in the case where they yield the relative entropy.\\
The fact that they easily deliver $\lim_{\alpha\to1}\hat D_{\alpha}(\rho\|\sigma)$ made us step away from attempts to prove that they fulfill the data processing inequality. Also, we left any attempts to prove joint convexity to future work. We do however prove that the R\'enyi axioms are fulfilled, except from the generalized mean value axiom.\\
First, we need some preliminary notation. Let $\rho,\sigma\in\cS(\mathbb C^2)$ satisfy $[\rho,\sigma]\neq0$. Since $[\rho,\sigma]\neq0$ it is clear that the set $V:=\{a\cdot \rho+b\cdot \sigma+c\cdot Id:a,b,c\in\mathbb R\}$ is a two-dimensional real vector space and its intersection with the Bloch sphere defines a convex subset of the latter. We may for sake of simplicity assume that $B$ is the eigenbasis of $\rho$ such that
\begin{align}
\rho=\left(\begin{array}{ll}\tfrac{1}{2}(1+z)&0\\0&\tfrac{1}{2}(1-z)\end{array}\right)
\end{align}
and $z\in(1/2,1]$. We may further assume without loss of generality that the representation of $\sigma$ is such that it has only real and positive entries. Both of these assumptions translate to unitary actions which depend on $\rho$ and $\sigma$. Then, the unitary transformations
\begin{align}
U_\varphi:=\left(\begin{array}{ll}\cos\varphi&-\sin\varphi\\ \sin\varphi&\cos\varphi\end{array}\right),\qquad\varphi\in[-\pi,\pi]
\end{align}
yield a set of unitary transformations which rotates only the hyperplane defined by $\rho$ and $\sigma$. The value $\varphi'\geq0$ at which we get $U_{-\varphi'}\sigma U_{-\varphi'}^\dag=a\cdot Id+b\cdot \rho$ can be used to define the set $\{U_{\varphi'\cdot t}\}_{t\in[0,1]}$ which satisfies that $U_0\rho U_0^\dag$ is diagonal and $U_1\sigma U_1^\dag$ is diagonal in the basis $B$ (the computational basis). If $[\rho,\sigma]=0$ we set $U_t=Id$ for all $t\in[0,1]$. If $d=1$ the convention $U_t=Id_{\mathbb C}=1$ applies as well. We are ready for a definition:
\begin{definition}\label{def:relative-entropies}
Let $d\in[2]$. To any $\rho\in M_d$ satisfying $\rho\geq0$ and $\sigma\in\cS(\mathbb C^2)$, set $\bar\rho:=\tr\{\rho\}^{-1}\cdot\rho$ and let $U_t:=U_t(\bar\rho,\sigma)$ be the set of unitary transformations which arise from $\bar\rho$ and $\sigma$ as described above. This defines a set $\{R_t\}_{t\in[0,1]}$ of relative entropy like functionals via
\begin{align}
R_t(\rho\|\sigma):=\Delta_{U_t^\dag}(\pinch(U_t\bar\rho U_t^\dag),\spec(\bar\rho), \tr\{\rho^{-1}\}\cdot\sigma).
\end{align}
\end{definition}
\begin{remark}
Theorem \ref{main-theorem} ensures that the definition does not only lead to trivial concatenations of rotations followed by unitary transformations, since $R_0=\hat D_1$ and $R_1=D$ are different functions. The normalization factor $\tr\{\rho^{-1}\}$ in front of $\sigma$ enables one to prove the order axiom.
\end{remark}
The structure of the functions we defined so far delivers operationally meaningful quantities right from the start, as they describe the asymptotic scaling of the probability that certain tests yield specific outcomes given that a system is in state $\sigma^{\otimes n}$ and $n$ is large.
\end{section}
\begin{section}{Proofs}
We now give the proofs of our theorems, in order of appearance.
\begin{proof}[Proof of Theorem \ref{main-theorem}]
The proof of statement $2$ is implicit in \cite{noetzel-hypothesis} and what is left to do is giving the proof of statement $1$. Any of the representations $V_{f,\lambda}$ is irreducible if $f(i)=\lambda_{\tau(i)}$ for some permutation $\tau\in S_d$. This can be seen as follows: Denote the set of all tableau $T$ of shape $\lambda$ by $\mathbb T_\lambda$. Then statement $1$ can be seen to hold true as follows: Remember that each $V_f$ is invariant under $\mathbb B$, so that for each $T\in\mathbb T_\lambda$ we have $E_Tv\in V_f$ whenever $v\in V_f$, where
\begin{align}
E_{T}:=\sum_{\upsilon\in C_T}\sum_{\tau\in R_T}\sgn(\upsilon)\mathbb B(\upsilon\circ\tau).
\end{align}
is the Young symmetrizer corresponding to the tableau $T$, $R_T$ is the set of permutations which permute only the elements in each row of $T$ amongst each other and $C_T$ permutes only the elements in the columns of $T$. By \cite[Chapter 5.5]{sternberg} (replace the object $T_n\mathbb C^d:=(\mathbb C^d)^{\otimes n}$ there with $V_f$) it holds that for every fixed $T\in\mathbb T_\lambda$ the dimension of the vector space $\linspan\{E_Tv:v\in V_f\}$ gives the multiplicity of $F_\lambda$ within $V_f$. Now let for sake of simplicity $f=\lambda$. It is clear that for the vector $v=\otimes_{i=1}^n e^{\otimes f(i)}$ and $T$ be 'the' standard tableau with numbers $1,\ldots,n$ filled in starting from left to right in the first row, then carrying on from left to right in the second row, and so on.\\
Then $\tilde v:=E_Tv\neq0$. Thus $\dim V_{\lambda,\lambda}>0$. Now take any other product vector $w=\otimes_{i=1}^ne_{x_i}$ where $x^n\in T_f$. There is at least one column (say the first) having at least two equal entries (for example it holds that both $x_1=1$ and $x_{\lambda_1+1}=1$). This statement is valid as well for every $w_\tau:=\mathbb B(\tau)w$ whenever $\tau\in R_T$, only the position of the specific column changes. Take the permutation $\pi=(1,\lambda_1+1)\in S_n$ which interchanges the elements $x_1$ and $x_{\lambda_1+1}$ for every $x^n\in[d]^n$. It holds $\mathbb B(\pi)w=w$. On the other hand $\pi\in C_T$ and $\sgn(\pi)=-1$. For each $\tau\in R_T$, let $\pi_\tau\in C_T$ be a corresponding permutation that satisfies $\mathbb B(\pi_\tau)w_\tau=w_\tau$ and $\sgn(\pi_\tau)=-1$. It follows
\begin{align}
E_Tw&=\sum_{\tau\in R_T}\sum_{\upsilon\in C_T}\sgn(\upsilon)\mathbb B(\upsilon)w_\tau\\
&=\sum_{\tau\in R_T}\sum_{\upsilon\in C_T}\sgn(\upsilon)\mathbb B(\upsilon)\cdot\mathbb B(\pi_\tau) w_\tau\\
&=-\sum_{\tau\in R_T}\sum_{\upsilon\in C_T}\sgn(\upsilon\circ\tau)\mathbb B(\upsilon)\cdot\mathbb B(\pi_\tau) w_\tau\\
&=-\sum_{\tau\in R_T}\sum_{\upsilon\in C_T}\sgn(\upsilon)\mathbb B(\upsilon) w_\tau\\
&=-E_Tw,
\end{align}
so that $\dim V_{\lambda,\lambda}=1$ follows. The argument is independent under a transformation $f\mapsto f\circ\tau$ whenever $\tau\in S_d$, so that all the representations $V_{f,\lambda}$ for which $f(i)=\lambda_{\tau(i)}$ for all $i\in[d]$ holds true for some $\tau\in S_d$ are irreducible.\\
Now we connect our first observation to a trick that we shall use more often in what follows:\\
Let $V$ be an irreducible subspace of the symmetric group. Then for every $0\neq v\in V$ we have $A_v:=\sum_{\tau\in S_n}\frac{1}{n!}\mathbb B(\tau)|v\rangle\langle v|\mathbb B(\tau^{-1})=c\cdot P_V$ for some $c=c(v)>0$ and the orthogonal projection $P_v$ onto $V$. This is seen as follows: note first that $A_v\neq 0$ whenever $v\neq0$. Furthermore each $A_v$ is invariant under permutations. By Schur's Lemma (see e.g. \cite[Chapter 2.3]{sternberg}) it follows that $A_v=c(v)\cdot P_V$.\\
Moreover, by taking the trace we see that $\|v\|_2^2=c\cdot\tr\{P_V\}$. This implies that for every $\sigma\in\mathbb C^d$ and $v\in V$ we have
\begin{align}
\tr\{P_V\sigma^{\otimes n}\}&=\frac{\tr\{P_V\}}{\|v\|^2_2}\tr\{|v\rangle\langle v|\sigma^{\otimes n}\}\\
&=\frac{\dim(V)}{\|v\|^2_2}\langle v,\sigma^{\otimes n}v\rangle.
\end{align}
Thus all that is left to do in this case is to construct one vector $v$ within $V_{f,\lambda}$ and calculate its norm as well as $\langle v,\sigma^{\otimes n}v\rangle$.\\
This task again is straightforward since we may just use the standard tableau $T$ that we defined already and the corresponding Young symmetrizer $E_{T}$. Applying this symmetrizer to the vector $\otimes_{i=1}^de_i^{\otimes f(i)}$ yields (without loss of generality $f=f^\downarrow$):
\begin{align}
v&:=E_{T}\otimes_{i=1}^de_i^{\otimes f(i)}\\
&=|R_T|\sum_{\tau\in C_{T}}\sgn(\tau)\mathbb B(\tau)\otimes_{i=1}^de_i^{\otimes f(i)}\\
&=|R_T|\mathbb B(\tau')\otimes_{i=1}^d\left(\sum_{\tau\in S_{i}}\sgn(\tau)\mathbb B(\tau)e_1\otimes\ldots e_i\right)^{\otimes(\lambda_i-\lambda_{i+1})},
\end{align}
where $\tau'$ is a suitably defined permutation. For any $k\in[d]$ we may now define $v_k\in(\mathbb C^d)^{\otimes k}$ by
\begin{align}
v_k:=\sum_{\tau\in S_{[k]}}\sgn(\tau)\mathbb B(\tau) e_1\otimes\ldots\otimes e_k,
\end{align}
a shorthand that allows us to write
\begin{align}
v&=|R_T|\mathbb B(\tau')\otimes_{i=1}^dv_i^{\otimes(\lambda_i-\lambda_{i+1})}.
\end{align}
In order to get a lower bound on the norm of $v$ we first note that each $v_i$ satisfies $\|v_i\|_2^2=i!$, so that
\begin{align}
\|v\|_2^2&=|R_T|^2\prod_{i=1}^d(i!)^{\lambda_i-\lambda_{i+1}}.
\end{align}
Another important ingredient is the equality
\begin{align}
\sum_{\tau\in S_k}\sum_{\upsilon\in S_k}\sgn(\tau)\sgn(\sigma)\mathbb B(\tau)\mathbb B(\upsilon)=k!\sum_{\tau\in S_k}\sgn(\tau)\mathbb B(\tau)
\end{align}
which lets us conclude that
\begin{align}
\langle v_k,\sigma^{\otimes k}v_k\rangle=k!\cdot\det(\sigma_{1:k,1:k}).
\end{align}
We are finally able to compute
\begin{align}
\tr\{P_{f,\lambda}\sigma^{\otimes n}\}&=\frac{\dim(V)}{\|v\|_2^2}|R_T|^2\prod_{i=1}^d\langle v_i,\sigma^{\otimes i}v_i\rangle^{\lambda_i-\lambda_{i+1}}\\
&=\frac{\dim(V)}{\|v\|_2^2}|R_T|^2\prod_{i=1}^d(i!\det(\sigma_{1:i,1:i}))^{\lambda_i-\lambda_{i+1}}\\
&=\frac{\dim(V)}{\|v\|_2^2}|R_T|^2\left(\prod_{i=1}^d(i!)^{\lambda_i-\lambda_{i+1}}\right)\left(\prod_{j=1}^d\det(\sigma_{1:j,1:j}))^{\lambda_j-\lambda_{j+1}}\right)\\
&=\frac{\dim(V)}{\|v\|_2^2}|R_T|^2\left(\prod_{i=1}^d(i!)^{\lambda_i-\lambda_{i+1}}\right)\left(\prod_{j=1}^d\det(\sigma_{1:j,1:j}))^{\lambda_j-\lambda_{j+1}}\right)\\
&=\dim(V)\prod_{j=1}^d\det(\sigma_{1:j,1:j}))^{\lambda_j-\lambda_{j+1}}\\
&=\mathrm{pl}(n)2^{n\cdot H(\bar\lambda)}2^{n\cdot\sum_{i=1}^d(\bar\lambda_i-\bar\lambda_{i+1})\log\det(\sigma_{1:i,1:i})}\\
&=\mathrm{pl}(n)2^{n\cdot H(\bar\lambda)+\sum_{i=1}^d(\bar\lambda_i-\bar\lambda_{i+1})\log\det(\sigma_{1:i,1:i})}\\
&=\mathrm{pl}(n)2^{n\cdot S(\rho)+\sum_{i=1}^d(\mu_i-\mu_{i+1})\log\det(\sigma_{1:i,1:i})+\epsilon(n)}\\
&=\mathrm{pl}(n)2^{n\cdot(D(\rho\|\hat\sigma)+\epsilon(n))}
\end{align}
where $\hat\sigma$ is a nonnegative matrix which is simultaneously diagonal with $\rho$ and is defined via its diagonal entries $\hat\sigma_{ii}=\det(\sigma_{1:i,1:i})$ and $\lim_{n\to\infty}\epsilon(n)=0$. In \cite{audenaert-datta} (see Theorem 2 with the respective parameter $r$ of the theorem set to $r=-1$ and equation $(24)$ there) it has been proven that $D(\rho\|\hat\sigma)=\lim_{\alpha\to1}\hat D_\alpha(\rho\|\sigma)$, so that ultimately we have
\begin{align}
-\lim_{n\to\infty}\frac{1}{n}\log\tr\{P_{f,\lambda}\sigma^{\otimes n}\}=\lim_{\alpha\to1}\hat D_\alpha(\rho\|\sigma)
\end{align}
as desired.
\end{proof}
\ \\\\
We now start to investigate the scalar products
\begin{align}
\tr\{A^{\otimes n}P_{f,\lambda}A^{\dag\otimes n}P_{f',\lambda'}\}.
\end{align}
It is generally clear that $\lambda=\lambda'$ has to hold, so that the questions we pose get reduced to the evaluation of quantities of the form $\tr\{A^{\otimes n}P_gA^{\dag\otimes n}P_\lambda P_f\}$.
\begin{proof}[Proof of Theorem \ref{theorem:scalar-products-for-arbitrary-d}]
We now work with an arbitrary $d\in\nn$. We write
\begin{align}
\tr\{A^{\otimes n}P_gA^{\dag\otimes n}P_{\lambda,\lambda}\}&=\frac{\tr\{P_{\lambda,\lambda}\}}{\|v\|_2^2}\tr\{|v\rangle\langle v|P_g\}
\end{align}
as before, and again $v$ takes the form
\begin{align}
v&=|R_T|\bigotimes_{i=1}^d\left(\sum_{\tau\in S_{i}}\sgn(\tau)\mathbb B(\tau)e_1\otimes\ldots e_i\right)^{\otimes(\lambda_i-\lambda_{i+1})}.
\end{align}
Again we set $v_k:=\tfrac{1}{\sqrt{k}}\sum_{\tau\in S_k}\sgn(\tau)\mathbb B(\tau)e_1\otimes\ldots\otimes e_k$. The asymptotic behaviour of the function $(f,\lambda)\mapsto\tfrac{1}{n}\log\tr\{P_{f,\lambda}\}$ is known to equal that of $\lambda\mapsto\tfrac{1}{n}\log\tr\{P_\lambda\}$ for all $f$ and $\lambda$ satisfying $K_{f,\lambda}>0$, so that what is left to do is the following: We have to calculate
\begin{align}
\tr\{A^{\otimes n}P_gA^{\dag\otimes n}\bigotimes_{i=1}^d|v_i\rangle\langle v_i|^{\otimes(\lambda_i-\lambda_{i+1})}\}&=\sum_{g_1+\ldots+g_d=g}\prod_{i=1}^d\tr\{A^{\dag\otimes(\lambda_i-\lambda_{i+1})}P_{g_i}A^{\dag\otimes(\lambda_i-\lambda_{i+1})}|v_i\rangle\langle v_i|^{\otimes(\lambda_i-\lambda_{i+1})}\},
\end{align}
and again we now have to dive into calculating, for every $m\in\nn$ and $h\in \mathbb T_m$ and $t\in \mathbb T_{m'}$, as well as for every $k$, quantities like $\tr\{A^{\otimes k\cdot m}P_hA^{\dag\otimes k\cdot m}|v_k\rangle\langle v_k|^{\otimes m}\}$. This task needs some additional notation. Let $m=m'\cdot k$ for some natural numbers $m'$ and $k$. Then, we set
\begin{align}\label{eqn:definition-of-F^h_t}
F^h_t:=\left\{\begin{array}{ll}1,&\qquad\mathrm{if}\ \ \sum_{(i_1,\ldots,i_k)\in[d]^k_j} t(i_1,\ldots,i_k)=h(j)\ \forall\ j\in[d],\\
0,&\qquad\mathrm{else}\end{array}\right..
\end{align}
It then holds that
\begin{align}\label{eqn:type-bound}
\eins_{T_h}&\geq\sum_{t\in\mathbb T_{m'}([d]^k)}F^h_t\cdot\eins_{\{h\}}.
\end{align}
Moreover, those types $t$ which do not occur on the right hand side of inequality (\ref{eqn:type-bound}) but only on the left are exactly those which have $t((i_1,\ldots,i_k))>0$ for some choice $(i_1,\ldots,i_k)$. Such types however satisfy
\begin{align}
\langle v_k,A^{\otimes k}e_{i_1}\otimes\ldots\otimes e_{i_k}\rangle=0
\end{align}
by symmetry of $v_k$. This justifies (actually it does so only in the second row of below chain of estimates so one has to read from there both back- and forwards) that we write
\begin{align}
\tr\{A^{\otimes m\cdot k}P_hA^{\dag\otimes m\cdot k}|v_k\rangle\langle v_k|^{\otimes m'}\}&=\sum_{t}F^h_t\tr\{A^{\otimes k}P_tA^{\dag\otimes k}|v_k\rangle\langle v_k|^{\otimes m'}\}\\
&=\sum_tF^h_t\prod_{(i_1,\ldots,i_k)}\langle v_k,A^{\otimes k}e_{i_1}\otimes\ldots\otimes e_{i_k}\rangle^{t(i_1,\ldots,i_k)}\\
&=\sum_tF^h_t2^{m'\cdot\sum_{(i_1,\ldots,i_k)}\tfrac{1}{m'} t(i_1,\ldots,i_k)\log(|\langle v_k,e_{i_1}\otimes\ldots e_{i_k}\rangle|^2)}\\
&\leq \mathrm{pl}(n)2^{m'\cdot \max_{t}F^h_t(H(\bar t)+\sum_{(i_1,\ldots,i_k)}\bar t(i_1,\ldots,i_k)\log(|\langle v_k,e_{i_1}\otimes\ldots e_{i_k}\rangle|^2))}
\end{align}
Upon normalization, the definition of $F^h_t$ translates into the set of probability distributions on $[d]^k$ which we defined in the introduction: The set
\begin{align}
\mathfrak P(q,k)=\left\{p\in\mathfrak P([d]^k):\begin{array}{l}N(i|(i_1,\ldots,i_k))>1\ \Rightarrow p((i_1,\ldots,i_k))=0\\ p([d]^k_i)=k\cdot q(i)\ \forall\ i\in[d]\end{array}\right\}.
\end{align}
It is this set that determines the asymptotic behaviour we are after: namely, for $\lim_{n\to\infty}\tfrac{1}{n}h^{(n)}=q$ it holds that
\begin{align}
\lim_{n\to\infty}\frac{1}{k\cdot n}\log\tr\{P_{h^{(k\cdot n)}}|v_k\rangle\langle v_k|^{\otimes n}\}=\frac{1}{k}\min_{p\in\mathfrak P(q,k)}D(p\| p_k)
\end{align}
with $p_k\in\mathfrak P([d]^k)$ being defined via $p_k(i_1,\ldots,i_k):=|\langle e_{i_1}\otimes\ldots\otimes e_{i_k},A^{\otimes k}v_k\rangle|^2$ for all $(i_1,\ldots,i_k)\in[d]^k$. For sequence of partitions $g_1^{(n)}+\ldots+g_d^{(n)}=g^{(n)}$ (for each $g_k^{(n)}$ we have $n_k:=\sum_{j=1}^dg_k^{(n)}(j)=k\cdot(\lambda_k-\lambda_{k+1})$ and $\lim_{n\to\infty}n_k/n=\hat s(k)$ for all $k=1,\ldots,d$) with respective limiting distributions $q,q_1,\ldots,q_d$ we get
\begin{align}
\lim_{n\to\infty}\frac{1}{n}\log\tr\{P_{g^{(n)}}P_{\lambda^{(n)},\lambda^{(n)}}\}&=\lim_{n\to\infty}\frac{1}{n}\log \tr\{P_{\lambda^{(n)}}\}\tr\left\{\bigotimes_{k=1}^dP_{g_k^{(n)}}|v_k\rangle\langle v_k|^{\otimes\left(\lambda_k^{(n)}-\lambda_{k+1}^{(n)}\right)}\right\}\\
&=H(s)+\sum_{k=1}^d\lim_{n\to\infty}\frac{1}{n}\log\tr\left\{P_{g_k^{(n)}}|v_k\rangle\langle v_k|^{\otimes\left(\lambda_k^{(n)}-\lambda_{k+1}^{(n)}\right)}\right\}\\
&=H(s)-\sum_{k=1}^d\hat s(k)\frac{1}{k}\min_{p\in\mathfrak P(q_k,k)}D(p\| p_k)
\end{align}
Naturally, this leads us to the formula
\begin{align}
\lim_{n\to\infty}\frac{1}{n}\log\tr\{P_{g^{(n)}}P_{\lambda^{(n)},\lambda^{(n)}}\}&=\lim_{n\to\infty}\frac{1}{n}\log \tr\{P_{\lambda^{(n)}}\}\tr\left\{P_{g^{(n)}}\bigotimes_{k=1}^d|v_k\rangle\langle v_k|^{\otimes\left(\lambda_k^{(n)}-\lambda_{k+1}^{(n)}\right)}\right\}\\
&=H(s)-\min_{W(\hat s)=q}\sum_{k=1}^d\hat s(k)\min_{p\in\mathfrak P(W(\delta_k),k)}\tfrac{1}{k}\cdot D(p\|p_k)
\end{align}
which is valid for all asymptotic shapes $s=\lim_{n\to\infty}\tfrac{1}{n}\lambda^{(n)}$ and where $\hat s(k):=\left(s(k)-s(k+1)\right)\cdot k$ and $s(k+1):=0$. Note that
\begin{align}
\sum_{k=1}^d\hat s(k)&=\sum_{k=1}^ds(k)\cdot k-\sum_{k=1}^ds(k+1)\cdot k\\
&=\sum_{k=1}^ds(k)\cdot k-\sum_{k=2}^ds(k)\cdot(k-1)\\
&=s(1)+\sum_{k=2}^ds(k)\\
&=1.
\end{align}
Let now $A=U\in M_d$ be a unitary matrix and $(q_k)_{k=1}^d$ and $(p(\cdot|k))_{k=1}^d$ be such that
\begin{align}
\Theta(q,s,U)&=H(s)-\sum_{k=1}^d\hat s(k)\frac{1}{k}D(p(\cdot|k)\| p_k).
\end{align}
By Pinsker's inequality and convexity of $x\mapsto x^2$ we get
\begin{align}
\Theta(q,s,U)&\leq H(s)-2\sum_{k=1}^d\hat s(k)\tfrac{1}{k}\|p(\cdot|k)-p_k\|^2\\
&\leq H(s)-2\left(\sum_{k=1}^d\hat s(k)\tfrac{1}{\sqrt{k}}\|p(\cdot|k)-p_k\|\right)^2.
\end{align}
It is trivially true that $-\tfrac{1}{k}\leq-\tfrac{1}{d}$ for all $1\leq k\leq d$, so that by monotonicity of the square root we get
\begin{align}
\Theta(q,s,U)&\leq H(s)-2\sum_{k=1}^d\hat s(k)\tfrac{1}{\sqrt{k}}\|p(\cdot|k)-p_k\|^2\\
&\leq H(s)-\tfrac{2}{d}\left(\sum_{k=1}^d\hat s(k)\|p(\cdot|k)-p_k\|\right)^2.
\end{align}
The distributions $q_1,\ldots,q_d$ and $p(\cdot|1),\ldots,p(\cdot|d)$ fulfill the assumptions of Lemma \ref{lemma:norm-estimate}, so that by convexity of $\|\cdot\|$ we get
\begin{align}
\Theta(q,s,U)&\leq H(s)-\tfrac{2}{d}\|q-\tilde q\|^2,
\end{align}
where $\tilde q(i):=\sum_{k=1}^d\hat s(k)\sum_{l=1}^k\tfrac{|u_{il}|^2}{k}$ for all $i\in[d]$. Of course then, $\Theta(q,s,A)$ attains its maximum $H(s)$ when $q=\tilde q$. This is the only maximum, since the function $q\mapsto\Theta(q,s,U)$ is convex: Let $\lambda\in[0,1]$ and set $\lambda':=1-\lambda$. Then for arbitrary $q,q'\in\mathfrak P([d])$ satisfying $s\succeq q$ and $s\succeq q'$ we have
\begin{align}
\Theta(\lambda q+\lambda'q',s,U)&\leq\min_{W(\hat s)=q}\min_{W'(\hat s)=q'}\sum_{k=1}^d\hat s(k)\frac{1}{k}\min_{p\in(\lambda W+\lambda'W')(\delta_k)}D(p\|p_k)\\
&\leq \min_{W(\hat s)=q}\min_{W'(\hat s)=q'}\sum_{k=1}^d\hat s(k)\frac{1}{k}\min_{p\in W(\delta_k)}\min_{p'\in W(\delta_k)}D(\lambda p+\lambda' p'\|p_k)\\
&\leq \min_{W(\hat s)=q}\min_{W'(\hat s)=q'}\sum_{k=1}^d\hat s(k)\frac{1}{k}\min_{p\in W(\delta_k)}\min_{p'\in W(\delta_k)}(\lambda D(p\|p_k)+\lambda'D(p'\|p_k)\\
&=\min_{W(\hat s)=q}\min_{W'(\hat s)=q'}\sum_{k=1}^d\hat s(k)\frac{1}{k}\lambda\min_{p\in W(\delta_k)}D(p\|p_k)+\lambda'\min_{p'\in W(\delta_k)}D(p'\|p_k)\\
&=\min_{W(\hat s)=q}\min_{W'(\hat s)=q'}\left(\sum_{k=1}^d\tfrac{\hat s(k)}{k}\lambda\min_{p\in W(\delta_k)}D(p\|p_k)+\sum_{k=1}^d\tfrac{\hat s(k)}{k}\lambda'\min_{p'\in W(\delta_k)}D(p'\|p_k)\right)\\
&=\lambda \min_{W(\hat s)=q}\sum_{k=1}^d\tfrac{\hat s(k)}{k}\min_{p\in W(\delta_k)}D(p\|p_k)+\lambda'\min_{W'(\hat s)=q'}\sum_{k=1}^d\tfrac{\hat s(k)}{k}\min_{p'\in W(\delta_k)}D(p'\|p_k)\\
&=\lambda\Theta(p,s,U)+\lambda'\Theta(p',s,U).
\end{align}
\end{proof}
We now give the proofs of our two additional Lemmata. The asymptotic estimates for $\tr\{A^{\otimes n}P_{f,\lambda}A^{\dag\otimes n}P_{f',\lambda'}\}$, although they give a rather cumbersome impression, naturally introduced the distributions $p_k\in\mathfrak P([d]^k)$ defined by $p_l(i_1,\ldots,i_l):=|\langle e_{i_1}\otimes e_{i_k},U^{\otimes k}v_k\rangle|^2$. These again make it interesting to look at lower bounds on the exponent in terms of norms, which turn out to deliver a satisfying intuition once Lemma \ref{lemma:cauchy-binet-type-formula} holds. We will now prove this Lemma.
\begin{proof}[Proof of Lemma \ref{lemma:cauchy-binet-type-formula}]
Let us assume that $j=1$ holds. It will become evident from our proof that this is without loss of generality. We will have to consider appropriate submatrices of $U$, that are defined entrywise as follows.
\begin{align}
U[(i_l)_{l=1}^{k},(j_l)_{l=1}^{k}]_{mn}:=u_{i_mj_n}.
\end{align}
For $k$ a natural number, the letter $\mathbf k$ denotes the string $(1,\ldots,k)$. For $m\in[k]$, $\mathbf k\backslash m$ denotes the string $(1,\ldots,m-1,m+1,k)$. We let $[d](2,k)$ be the set of all strings of length $k$ with elements taken from $\{1,\ldots,d\}$ \emph{without repetition}. Using these matrices will allow us to employ first Laplace's formula, then the Cauchy-Binet formula followed by the Sherman-Morisson formula. Together with the fact that $U$ is a unitary matrix, this will lead to the desired result.\\
Now, we will write above sum as a sum over determinants. This will allow us to apply the Cauchy-Binet Formula - but first we have to rewrite our form slightly in order to see the determinants.
\begin{align}
\sum_{[d]^k_1}|\langle Ue_{i_1}\otimes\ldots Ue_{i_{k}},v_k\rangle|^2&=\sum_{[d]^k_1}|\langle Ue_{i_1}\otimes\ldots Ue_{i_{k}},\frac{1}{\sqrt{k}}\sum_{\tau\in S_{k}}\sgn(\tau)e_{\tau(1)}\otimes\ldots\otimes e_{\tau(k)}|^2\\
&=\frac{1}{k}\sum_{[d]^k_1}|\sum_{\tau\in S_{k}}\sgn(\tau)u_{i_1\tau(1)}\otimes\ldots\otimes u_{i_{k}\tau(k)}|^2\\
&=\frac{1}{k}\sum_{[d]^k_1}|\mathrm{det}([U,(i_1,\ldots,i_{k}),\mathbf k])|^2.
\end{align}
It is clear that the terms in above sum are invariant under permutations. The function $(i_1,\ldots,i_k)\mapsto|\langle Ue_{i_1}\otimes\ldots Ue_{i_{k}},v_k\rangle|^2$ is also designed such that $|\langle Ue_{i_1}\otimes\ldots Ue_{i_{k}},v_k\rangle|^2=0$ whenever $i_m=i_n$ for some $m\neq n$. This implies that it suffices to consider those terms where $1$ stands in the first place and $(1,i_2,\ldots,i_k)$ form an index set (meaning that $(i_2,\ldots,i_k)\in[d](2,k)$). We thus get the formula
\begin{align}
\sum_{[d]^k_1}|\langle Ue_{i_1}\otimes\ldots Ue_{i_{k}},v_k\rangle|^2&=\sum_{(i_2,\ldots,i_k)\in[d](2,k)}|\mathrm{det}([U,(1,i_2\ldots,i_{k}),\mathbf k])|^2.
\end{align}
Since only the columns $1$ to $k$ enter our calculations, let us consider $U$ as a $d\times k$ matrix from now on, with the transposed matrix $U^\top$ being a $k\times d$ matrix. $\overline U$ denotes the matrix having the complex conjugate entries of $U$. We will now apply Laplace's formula (twice), followed by the Cauchy-Binet formula \cite[Chapter 0.8.5]{horn-johnson} (set $r=k-1$ in the book):
\begin{align}
\sum_{[d]^k_1}&|\langle Ue_{i_1}\otimes\ldots Ue_{i_{k}},v_k\rangle|^2
=\sum_{(i_2,\ldots,i_k)\in[d](2,k)}|\mathrm{det}([U,(1,i_2\ldots,i_{k}),\mathbf k])|^2\\
&=\sum_{(i_2,\ldots,i_{k})\in[d](2,k)}\sum_{m,n=1}^{k}(-1)^{m+n}u_{1m}\overline{u_{1n}}\mathrm{det}([U,(i_2\ldots,i_{k}),\mathbf k\backslash m])\overline{\mathrm{det}([U,(i_2\ldots,i_{k}),\mathbf k\backslash n ])}\\
&=\sum_{(i_2,\ldots,i_{k})\in[d](2,k)}\sum_{m,n=1}^{k}(-1)^{m+n}u_{1m}\overline{u_{1n}}\mathrm{det}([U^\top,\mathbf k\backslash m,(i_2\ldots,i_{k})])\mathrm{det}([\overline U,(i_2\ldots,i_{k}),\mathbf k\backslash n])\\
&=\sum_{m,n=1}^{k}(-1)^{m+n}u_{1m}\overline{u_{1n}}\sum_{(i_2,\ldots,i_{k})\in[d](2,k)}\mathrm{det}([U^\top,\mathbf k\backslash m,(i_2\ldots,i_{k})])\mathrm{det}([\overline U,(i_2\ldots,i_{k}),\mathbf k\backslash n])\\
&=\sum_{m,n=1}^{k}(-1)^{m+n}u_{1m}\overline{u_{1n}}\mathrm{det}([U^\top,\mathbf k\backslash m,\mathbf d\backslash 1]\cdot[\overline U,\mathbf d\backslash 1,\mathbf k\backslash n]).
\end{align}
It looks tempting to re-apply the Laplace formula here, but the determinants are now being calculated on products of non-square matrices so that we have to find a different means of dealing with the above sum. Let us calculate above determinants. For a fixed pair $(m,n)\in[k]\times[k]$ we have to calculate the determinant of the $(k-1)\times(k-1)$ matrix $C(m,n)$ defined by
\begin{align}
C(m,n):=[U^\top,\mathbf k\backslash m,\mathbf d\backslash1]\cdot[\overline U,\mathbf d\backslash1,\mathbf k\backslash n]).
\end{align}
The entry of this matrix that corresponds to $i\in[k]\backslash \{m\}$ and $l\in[k]\backslash \{n\}$ is given by
\begin{align}
\sum_{r=2}^d[U^\top,\mathbf d\backslash1,\mathbf k\backslash m])_{ir}([\overline U,\mathbf k\backslash n,\mathbf d\backslash1])_{rl}&=\sum_{r=2}^du_{ri}\overline{u_{rl}}\\
&=\delta(i,l)-u_{1i}\overline{u_{1l}},
\end{align}
and it is exactly here that we use the fact that $U$ is a unitary matrix. Since all entries belonging to the $m$th row index and the $n$th column index are removed from $C(m,n)$, it is evident that $\det C(m,n)$ equals the $(m,n)$ minor of the $k\times k$ matrix $(\delta(i,l)-u_{1i}\overline{u_{1l}})_{il}$, which can equivalently be written as $M:=\eins-|u\rangle\langle u|\in M_{k}$ where $u=\sum_{i=1}^{k}u_{1i}e_i$. The determinant of $M$ is calculated as
\begin{align}
\det(M)=1-\sum_{i=1}^{k}|u_{1i}|^2
\end{align}
via \cite[Lemma 1.1]{ding-zhou}. This makes it useful to again apply Laplace's formula (twice, again), where $M_{mn}$ are the entries of $M$:
\begin{align}
1-\sum_{i=1}^{k}|u_{1i}|^2&=\det(M)\\
&=\frac{1}{k}\sum_{m,n=1}^{k}(-1)^{m+n}M_{nm}\mathrm{det} C(n,m)\\
&=\frac{1}{k}\sum_{m,n=1}^{k}(-1)^{m+n}(\delta(m,n)-u_{1m}\overline u_{1n})\mathrm{det} C(n,m).
\end{align}
It follows again from \cite[Lemma 1.1]{ding-zhou} that for every $m\in[k]$ we have $\mathrm{det}C(m,m)=1-\sum_{i=1}^k|u_{1i}|^2+|u_{1m}|^2$, so that
\begin{align}
1-\sum_{i=1}^{k}|u_{1i}|^2&=\left(1-\sum_{i=1}^k|u_{1i}|^2\right)+\frac{1}{k}\sum_{i=1}^k|u_{1i}|^2-\frac{1}{k}\sum_{m,n=1}^{k}(-1)^{m+n}u_{1m}\bar u_{1n}\mathrm{det} C(m,n)\\
&=\left(1-\sum_{i=1}^k|u_{1i}|^2\right)+\frac{1}{k}\sum_{i=1}^k|u_{1i}|^2-\frac{1}{k}\sum_{[d]_1^k}|\langle Ue_{i_1}\otimes\ldots\otimes Ue_{i_{k}},v_k\rangle|^2,
\end{align}
which is equivalent to saying that
\begin{align}
\sum_{i=1}^k|u_{1i}|^2=\sum_{[d]_1^k}|\langle Ue_{i_1}\otimes\ldots\otimes Ue_{i_{k}},v_k\rangle|^2
\end{align}
so that we have proven the desired formula.
\end{proof}
We now come to the proof of the estimate which originally motivated us to study the determinant equation which is the content of Lemma \ref{lemma:cauchy-binet-type-formula}.
\begin{proof}[Proof of Lemma \ref{lemma:norm-estimate}]
Let the preliminaries of the Lemma be fulfilled: We have that $s\in\mathfrak P^\downarrow([d])$ and define $\hat s\in\mathfrak P([d])$ by $\hat s(k)=\left(s(k)-s(k+1)\right)\cdot k$. Let there be distributions $q,q_1,\ldots,q_k\in\mathfrak P([d])$ such that $\sum_{k=1}^d\hat s(k)q_k=q$. Let further $U\in M_{d\times d}$ be a unitary matrix and $\tilde q\in\mathfrak P([d])$ be defined by $\tilde q(i):=\sum_{k=1}^d\hat s(k)\sum_{l=1}^k\tfrac{|u_{il}|^2}{k}$. For every $k\in[d]$ assume that there exists a $p(\cdot|k)\in\mathfrak P(q_k,k)$. Let finally $p_k\in \mathfrak P([d]^k)$ be defined by $p_k(i_1,\ldots,i_k):=|\langle e_{i_1}\otimes\ldots\otimes e_{i_k},U^{\otimes k}v_k\rangle|^2$. Let $[d]^k_{NR}:=\cup_{q\in\mathfrak P([d])}\mathfrak P(q,k)$. These are those sequences of length $n$ with elements taken from $[d]$ that have no single element occurring twice ($NR$ means ``no repetitions''). It holds
\begin{align}
\sum_{k=1}^d\hat s(k)\|p(\cdot|k)-p_k\|&=\sum_{k=1}^d\hat s(k)\cdot\sum_{(i_1,\ldots,i_k)\in[d]^k}|p(i_1,\ldots,i_k|k)-p_k(i_1,\ldots,i_k)|\\
&=\sum_{k=1}^d\hat s(k)\cdot\sum_{(i_1\ldots i_k)\in[d]^k_{NR}}|p(i_1,\ldots,i_k|k)-p_k(i_1,\ldots,i_k)|.
\end{align}
Since no repetitions are allowed in above sum we can be sure that, for every $k\in[d]$, the respective sum over $[d]^k_{NR}$ can be split up into sums over subsets $[d]^k_i$ as follows:
\begin{align}
\sum_{(i_1\ldots i_k)\in[d]^k_{NR}}|p(i_1,\ldots,i_k|k)-p_k(i_1,\ldots,i_k)|&=\sum_{i=1}^d\frac{1}{k}\sum_{(i_1,\ldots,i_k)\in[d]^k_i}|p(i_1,\ldots,i_k|k)-p_k(i_1,\ldots,i_k)|,
\end{align}
where the fact that $[d]^k_i\cap[d]^k_j\neq\emptyset$ has been taken care of by the factor $\tfrac{1}{k}$. We can use this to reformulate the above sum as
\begin{align}
\sum_{k=1}^d\hat s(k)\|p(\cdot|k)-p_k\|&=\sum_{k=1}^d\hat s(k)\frac{1}{k}\cdot\sum_{(i_1,\ldots,i_k)\in[d]^k}|p(i_1,\ldots,i_k|k)-p_k(i_1,\ldots,i_k)|\\
&=\sum_{k=1}^d\left(s(k)-s(k+1)\right)\sum_{i=1}^d\sum_{(i_1,\ldots,i_k)\in[d]^k_i}|p(i_1,\ldots,i_k|k)-p_k(i_1,\ldots,i_k)|\\
&\geq\sum_{k=1}^d\left(s(k)-s(k+1)\right)\sum_{i=1}^d|\sum_{(i_1,\ldots,i_k)\in[d]^k_i}\left(p(i_1,\ldots,i_k|k)-p_k(i_1,\ldots,i_k)\right)|\\
&=\sum_{k=1}^d\left(s(k)-s(k+1)\right)\sum_{i=1}^d|p([d]^k_i|k)-\sum_{l=1}^k|u_{il}|^2|\\
&=\sum_{k=1}^d\left(s(k)-s(k+1)\right)\sum_{i=1}^d|k\cdot q_k(i)-k\cdot\sum_{l=1}^k\frac{|u_{il}|^2}{k}|\\
&\geq\sum_{i=1}^d|\sum_{k=1}^d\hat s(k)\cdot q_k(i)-\sum_{k=1}^d\hat s(k)\sum_{l=1}^k\frac{|u_{il}|^2}{k}|\\
&=\sum_{i=1}^d|q(i)-\tilde q(i)|\\
&=\|q-\tilde q\|.
\end{align}
\end{proof}
\end{section}

\begin{section}{Proofs for $d=2$}
\begin{proof}[Proof of Theorem \ref{theorem:characterization-of-DU}]
In order to go further with our investigation of the asymptotic behaviour of $\tr\{P_{f,\lambda}\sigma^{\otimes n}\}$ for arbitrary $f$ and $\lambda$ we unfortunately have to live with the restriction $d=2$.\\
In our case there is no difference: While it seems to be a rather involved task to obtain explicit formulas for the case $f\neq\lambda$ whenever $d\geq3$ we are well able to so when $d=2$. The reason for this is that in this case we always have $K_{f,\lambda}\in\{0,1\}$, so that each $V_{f,\lambda}$ is \emph{irreducible}! This can be seen as follows:\\
According to \cite[Chapter 5.5]{sternberg} (replace the object $T_n\mathbb C^d:=(\mathbb C^d)^{\otimes n}$ there with $V_f$), the multiplicity of $F_\lambda$ within $V_{f,\lambda}$ is given by $\linspan\{E_Tv:v\in V_f\}$ where $T\in\mathbb T_\lambda$ is any standard tableaux of shape $\lambda$.\\
Let again $T$ be ``the'' standard tableau with entries $T_{ij}=(j-2)\cdot i+(j-1)\cdot(\lambda_1+i)$ and $w=\otimes_{i=1}^ne_{x_i}$ for some $x^n\in T_f$. Then
\begin{align}
E_Tv&=\sum_{\tau\in R_T}\sum_{\upsilon\in C_T}\sgn(\upsilon)\mathbb B(\upsilon)\cdot\mathbb B(\tau)v.
\end{align}
Define $A_T:=\sum_{\upsilon\in C_T}\sgn(\upsilon)\mathbb B(\upsilon)$ and $B_T:=\sum_{\tau\in R_T}\mathbb B(\tau)$. For every $\tau\in R_T$, set $w_\tau:=\mathbb B(\tau)v$. For every $\tau\in R_T$ we see that $A_Tw_\tau=0$ holds if and only if $\tau x^n=(x_{\tau^{-1}(1)},\ldots,x_{\tau^{-1}(n)})$ satisfies $(\tau x^n)_i=(\tau x^n)_j$ for some pair $(i,j)$ where $1\leq i\leq\lambda_2$ and $\lambda_1+1\leq j\leq n$. It follows that $A_Tw_\tau=c(v,\tau)\tilde v$ for some set of non-negative numbers $\{c(v,\tau)\}_{\tau\in S_n}$ and a vector $\tilde v$ to be calculated more explicitly later. Thus
\begin{align}
E_Tw&=\left(\sum_{\tau\in R_T}c(v,\tau)\right)\tilde w,
\end{align}
proving that the multiplicity of $F_\lambda$ in $V_{f,\lambda}$ is at most one. We proceed with the calculation of $E_Tw$. Again, take the Young symmetrizer $E_{T_\lambda}$. Then for some constant $c'(f,\lambda)$ we get
\begin{align}
v&:=E_{T_{\lambda}}\otimes_{i=1}^2e_i^{\otimes f(i)}\\
&=c'(f,\lambda)\mathbb B(\tau')(\frac{1}{\sqrt{2}}v_2)^{\otimes(\lambda_1-\lambda_2)}\otimes v_{f-\lambda_2},
\end{align}
where $v_{f-\lambda_2}$ is defined only for those pairs $(f,\lambda)$ for which both $f(1)\geq\lambda_2$ and $f(2)\geq\lambda_2$ holds. In that case $f-\lambda_2$ defines a new type $g:=f-\lambda_2$ on $\{1,2\}^{\lambda_1-\lambda_2}$ so that it generally makes sense to define for an arbitrary $g\in \mathbb T_m$:
\begin{align}
v_g:=\frac{1}{|T_g|}\sum_{x^m\in T_g}e_{x_1}\otimes\ldots\otimes e_{x_m}.
\end{align}
The asymptotic scaling of $\tr\{P_{f,\lambda}\sigma^{\otimes n}\}$ is then conveniently calculated by starting with
\begin{align}
-\frac{1}{n}\log\tr\{P_{f,\lambda}\sigma^{\otimes n}\}&=-\frac{1}{n}\log\left(\frac{\dim V_\lambda}{c(f,\lambda)^2}c(f,\lambda)^2(\frac{1}{2}\langle v_2,\sigma^{\otimes 2}v_2\rangle^{\lambda_2}\langle v_{f-\lambda_2},\sigma^{\otimes(f-\lambda_2)}v_{f-\lambda_2}\rangle\right)\\
&=-\frac{1}{n}\left(\log\dim(V_\lambda)+\log(\tfrac{1}{2}\langle v_2,\sigma^{\otimes 2}v_2\rangle^{\lambda_2})+\log(\langle v_{f-\lambda_2},\sigma^{\otimes(f-\lambda_2)}v_{f-\lambda_2}\rangle)\right),
\end{align}
then calculating the limit of the three terms in the sum separately yields the desired result - but only if the limiting behaviour of the last one of them is known. We thus start with that part. Under the assumption that $(\frac{1}{n}f^{n})_{n\in\nn}$ converges, we call the limiting object $p:=\lim_{n\to\infty}\frac{1}{n}f^{n}$. For any given state $\sigma$ on $\mathbb C^d$ with matrix representation $\sigma=\sum_{i,j}\sigma_{ij}|e_i\rangle\langle e_j|$, it is then of interest to describe the limit
\begin{align}
\lim_{n\to\infty}\frac{1}{n}\log(\langle v_{f^{n}},\sigma^{\otimes n}v_{f^{n}}\rangle).
\end{align}
This can also be cast in to the form of the subspaces $V_{f,\lambda}$ by setting $\lambda=(n,0,0,\ldots,0)$. In order to have a more streamlined notation, we will drop the superscript $n$ in $f^{n}$ for now, then we can upper bound the limit as follows:
\begin{align}
\langle v_f,\sigma^{\otimes n}v_f\rangle&=\frac{1}{|T_f|}\sum_{x^n,y^n\in T_f}\prod_{i=1}^n(\sigma)_{x_i,y_i}\\
(\mathrm{for\ some\ fixed}\ y^n\in T_f)\ &=\sum_{x^n\in T_f}\prod_{i=1}^n(\sigma)_{x_i,y_i}\\
&=\sum_{x^n\in T_f}\prod_{i=1}^2(\sigma)_{i,j}^{N(i,j|x^n,y^n)}\\
&=\sum_{x^n\in T_f}\prod_{i=1}^2c_{ij}^{N(i,j|x^n,y^n)}
\end{align}
where we have set $c_{ij}:=|\sigma_{ij}|$. Now for every pair $x^n,y^n\in T_f$ it is clear that the numbers $N(i,j|x^n,y^n):=|\{k:x_k=i,\ y_k=j\}|$ satisfy $N(i,j)=N(j,i)$ for all $i,j\in\{1,2\}$, this being another peculiarity of the case $d=2$. Note that this implies
\begin{align}
\langle v_f,\sigma^{\otimes n}v_f\rangle&=\sum_{x^n\in T_f}\prod_{i=1}^2c_{ij}^{N(i,j|x^n,y^n)}.
\end{align}
Obviously we need some additional structure. This comes into play by decomposing the set $T_f$ according to
\begin{align}
T_f=\bigcup_{g_1,g_2}G^f_{g_1,g_2}T_{g_1}\times T_{g_2},
\end{align}
where each $T_{g_i}\subset[2]^{f(i)}$, and the above union is over \emph{disjoint} sets. The numbers $G^f_{g_1,g_2}\in\{0,1\}$ are defined in analogy to the Kronecker coefficients of the symmetric group, precisely speaking we set
\begin{align}\label{eqn:classical-kronecker-coefficient}
G^f_{g_1,g_2}:=\left\{\begin{array}{c c}
1,&\sum_{j=1}^2g_j(i)=f(i)\ \forall\ i\in[d],\\
0,&\mathrm{else}
\end{array}
\right.
\end{align}
The nice thing about this decomposition is that there are only \emph{polynomially} many (in $n$) different choices $(g_1,g_2)$ - more accurately, the number of such choices can be given a the loose upper bound $\mathrm{pl}(n):=(2n)^4$. This allows for a reasoning along the lines of the 'method of types':
\begin{align}
\langle v_f,\sigma^{\otimes n}v_f\rangle&=\sum_{g_1,g_2}G^f_{g_0,g_1}\prod_{i,j=1}^{2}|T_{g_j}|\cdot c_{ij}^{g_j(i)}\\
&\leq\mathrm{pl}(n)\max_{g_1,g_2:G^f_{g_1,g_2}>0}\prod_{i,j=1}^{d}|T_{g_j}|\cdot c_{ij}^{g_j(i)}.
\end{align}
As a consequence of an almost identical calculation it follows that
\begin{align}
\langle v_f,\sigma^{\otimes n}v_f\rangle&\geq\max_{g_1,g_2:G^f_{g_1,g_2}>0}\prod_{i,j=1}^{2}|T_{g_j}|\cdot c_{ij}^{g_j(i)}.
\end{align}
This demonstrates that the following holds: If $\lim_{n\to\infty}\tfrac{1}{n}f^{n}=p$, then
\begin{align}
\lim_{n\to\infty}\tfrac{1}{n}\log\langle v_{f^n},\sigma^{\otimes n}v_{f^n}\rangle=\lim_{n\to\infty}\tfrac{1}{n}\log\left( \max_{g_1,g_2:G^f_{g_1,g_2}>0}\prod_{i,j=1}^{2}|T_{g_j}|\cdot c_{ij}^{g_j(i)}\right)
\end{align}
holds whenever the two limits exist as well. We now translate our statements to a different regime by noting that $f(i)^{-2}|T_{g_i}|\leq 2^{nH(\bar g_i)}$, where $\bar g_j:=\frac{1}{f(j)}g_j$. We then estimate
\begin{align}
\langle v_f,\sigma^{\otimes n}v_f\rangle&\leq\max_{g_0,\ldots,g_d:G^f_{g_0,g_1}>0}2^{n\sum_{i,j=1}^{2}[\bar f(j)H(\bar g_j)+\bar f(j)\bar g_j(i)\log(c_{ij})]}\\
&=\max_{g_1,g_2:G^f_{g_1,g_2}>0}2^{-n\sum_{j=1}^{2}\bar f(j)D(\bar g_j\|c_{\cdot j})}.
\end{align}
If we now plug in the limiting behaviour $(\frac{1}{n}f^{n})_{n\in\nn}\to p$ and translate definition \ref{eqn:classical-kronecker-coefficient} to probability distributions by dividing through $n$, we end up with
\begin{align}
\langle v_f,\sigma^{\otimes n}v_f\rangle&\leq\sum_{g_1,g_2}G^f_{g_1,g_2}\prod_{j=1}^2\left(\sum_{x^{f(i)}\in T_{g_i}}\prod_{i=0}^{2}c_{ij}^{g_j(i)}\right)\\
&\leq\max_{W:W(p)=p}2^{-n\sum_{j=1}^{2}p(j)D(W(\delta_j)\|c_{\cdot j}|)}\\
&=2^{-n\min_{W:W(p)=p}\sum_{j=1}^{2}p(j)D(W(\delta_j)\|c_{\cdot j}|)},
\end{align}
and the symbol $W$ stands for the matrix $(w_{ij})_{i,j=1}^2$ with nonnegative entries and $w(1|i)+w(2|i)=1$ for $i=1,2$, and $W(p)=\sum_{i,j=1}^2w(i|j)p(j)\delta_i$ can be seen as application of the matrix $W$ to the vector $p=\sum_{i=1}^2p(i)\delta_i$, where $\delta_i(j):=\delta(i,j)$ are the usual Dirac distributions on $[2]$.\\
The calculation of a corresponding lower bound can be established with an almost identical reasoning, so that we obtain
\begin{align}
\lim_{n\to\infty}\tfrac{1}{n}\log\langle v_{f^n},\sigma^{\otimes n}v_{f^n}\rangle=-\min_{W:W(p)=p}\sum_{j=1}^2p(j)D(W(\delta_j)\||\sigma_{\cdot j}|).
\end{align}
We collect what we found so far in the following formula: For $\spec\rho=(\mu_1,\mu_2)$ with $\mu_1\geq\mu_2$ and pinching $\pinch\rho=(\nu_1,\nu_2)$ we have
\begin{align}
\Phi(\rho\|\sigma)=-S(\rho)+\mu_2\log\det(\sigma)+(\mu_1-\mu_2)\bar D((\tfrac{\nu_1-\mu_2}{\mu_1-\mu_2},\tfrac{\nu_2-\mu_2}{\mu_1-\mu_2})\|\sigma).
\end{align}
\end{proof}
We now turn our attention to the scalar products for the special case $d=2$, which allows for some stronger results.
\begin{proof}[Proof of Theorem \ref{theorem:scalar-products-for-d=2}]
Using the same tricks as in the proof of Theorem \ref{theorem:scalar-products-for-arbitrary-d} or Theorem \ref{main-theorem} we can write
\begin{align}
\tr\{P_{f,\lambda}A^{\otimes n}P_{f',\lambda'}A^{\dag\otimes n}\}&=\frac{\dim V_{f,\lambda}}{\|v\|^2}\tr\{A^{\otimes n}|v\rangle\langle v|A^{\otimes n}P_{f',\lambda'}\}\cdot\delta(\lambda,\lambda')
\end{align}
where $v$ is as defined below and with respect to a standard Young tableaux $T$ that we write e.g. for $n=15$ as
\begin{center}
\begin{tikzpicture}[scale=0.5]
\draw (0,3) to (8,3);
\draw (0,2) to (8,2);
\draw (0,1) to (6,1);
\draw(0,3) to (0,1);
\draw(1,3) to (1,1);
\draw(2,3) to (2,1);
\draw(3,3) to (3,1);
\draw(4,3) to (4,1);
\draw(5,3) to (5,1);
\draw(6,3) to (6,1);
\draw(7,3) to (7,2);
\draw(8,3) to (8,2);
\node at (0.5,2.5)  {$1$};
\node at (1.5,2.5)  {$3$};
\node at (2.5,2.5)  {$5$};
\node at (3.5,2.5)  {$7$};
\node at (4.5,2.5)  {$9$};
\node at (5.5,2.5)  {$11$};
\node at (6.5,2.5)  {$13$};
\node at (7.5,2.5)  {$15$};
\node at (0.5,1.5)  {$2$};
\node at (1.5,1.5)  {$4$};
\node at (2.5,1.5)  {$6$};
\node at (3.5,1.5)  {$8$};
\node at (4.5,1.5)  {$10$};
\node at (5.5,1.5)  {$12$};
\end{tikzpicture}
\end{center}
such that the role of the anti-symmetrizer $B_T=\sum_{\upsilon\in R_T}\sgn(\upsilon)\mathbb B(\upsilon)$ is to anti-symmetrize on the first $2\cdot\lambda_2$ blocks. We can then write (setting $g(i):=f(i)-\lambda_2$):
\begin{align}
v&=\frac{1}{2^{\lambda_2/2}\sqrt{|T_g|}}E_T\otimes_{i=1}^2e_i^{\otimes f(i)}\\
&=\frac{1}{2^{\lambda_2/2}\sqrt{|T_g|}}\sum_{\tau\in R_T}\sum_{\upsilon\in C_T}\sgn(\upsilon)\mathbb B(\upsilon\cdot\tau)\otimes_{i=1}^2e_i^{\otimes f(i)}\\
&=\frac{1}{2^{\lambda_2/2}}\left(\sum_{\upsilon\in C_T}\sgn(\upsilon)\mathbb B(\upsilon)\bigotimes_{i=1}^{\lambda_2}(e_1\otimes e_2)\right)\bigotimes v_g\\
&=v_2^{\otimes \lambda_2}\bigotimes v_g,
\end{align}
and what remains is to calculate the quantities $\|v\|_2^2$ and $\langle v,P_{f'}'v\rangle$. The former evaluates to $\|v\|_2^2=2^{\lambda_2}\cdot|T_g|$. We calculate the latter by exploiting the specific product structure of $v$ that we developed above. Going into details, we see that
\begin{align}
\langle v,A^{\otimes n}P_{f'}A^{\dag\otimes n}v\rangle&=\sum_{g_1,g_2}G^{f'}_{g_1,g_2}\tr\{A^{\otimes 2\cdot\lambda_2}P_{g_1}A^{\dag\otimes 2\cdot\lambda_2}|v_2^{\otimes \lambda_2}\rangle\langle v_2^{\otimes \lambda_2}| \}\cdot\tr\{A^{\dag\otimes(\lambda_1-\lambda_2)}P_{g_2}A^{\dag\otimes(\lambda_1-\lambda_2)}|v_g\rangle\langle v_g|\}
\end{align}
holds. It is now clear that we have to calculate, for every even natural number $n$ and type $g\in \mathbbm T_{n}$, the quantity $\tr\{P_gA^{\otimes n}v_2^{\otimes n/2}A^{\dag\otimes n}\}$ and for arbitrary $n\in\nn$ the number $\tr\{A^{\otimes n}P_gA^{\dag\otimes n}|v_g\rangle\langle v_g|\}$. This is done in the following. We define a function $p_{2,A}:[2]\times[2]\to\mathbb R_+$ via $p_{2,A}(i,j):=|\langle e_i\otimes e_j,A^{\otimes2}v_2\rangle|^2$ (if $A$ is a unitary matrix this is an element of $\mathfrak P([2]\times[2])$). Again we look, for every two types $g\in \mathbbm T_n([2])$ and $h\in \mathbbm T_{n/2}([2]\times[2])$, at the numbers $F^g_h$ that we defined in (\ref{eqn:definition-of-F^h_t}). These enable us to write
\begin{align}
\tr\{A^{\otimes n}P_gA^{\dag\otimes n}v_2^{\otimes n/2}\}&=\sum_{x^n\in T_g}\langle x^n,A^{\dag\otimes n}v_2^{\otimes n/2}A^{\otimes n}x^n\rangle\\
&=\sum_{x^n\in T_g}\prod_{i=1}^{n/2}|\langle e_{x_{2\cdot i}}\otimes e_{x_{2\cdot i+1}},A^{\dag\otimes 2}v_2\rangle|^2\\
&=\sum_{h}F^g_h|T_h|\prod_{i,j=1}^2|\langle e_i\otimes e_j,A^{\dag\otimes 2}v_2\rangle|^{2\cdot h(i,j)}\\
&\leq\mathrm{pl}(n)\max_h2^{n\tfrac{1}{2}(H(\bar h)+\sum_{i,j=1}^2\bar h(i,j)\log p_{2,A}(i,j))}.
\end{align}
The estimate can be carried out in the other direction as well such that we get, for every sequence $g^{(n)}$ such that $\lim_{n\to\infty}\tfrac{1}{n}g^{(n)}=q\in\mathfrak P([d])$ holds, the asymptotic relation
\begin{align}
\lim_{n\to\infty}\frac{1}{n}\log\tr\{A^{\otimes n}P_gA^{\dag\otimes n}|v_2^{\otimes n/2}\rangle\langle v_2^{\otimes n/2}|\}=-\tfrac{1}{2}\min_{r\in \mathfrak P(2,q)}D(r\|p_{2,A}),
\end{align}
and this does obviously imply that
\begin{align}
\Theta_2(A,q)=-\tfrac{1}{2}\min_{r\in \mathfrak P(2,q)}D(r\|p_{2,A}).
\end{align}
The other asymptotic quantity that needs to be calculated is still left open. Here, we proceed as follows: Define
\begin{align}
G^{g'\to g}_{g_1,g_2}:=\left\{\begin{array}{ll}1,&\sum_{i=1}^2g_j(i)=g(j),\ \ i=1,2\\
0,&\mathrm{else}\end{array}\right.,
\end{align}
where each $g_j\in\mathbbm T_{g'(j)}$ then it holds that with $p_{1,A}:[2]\times[2]\to\mathbb R_+$ via $p_{1,A}(i,j):=|\langle e_i,A^\dag e_j\rangle|^2$ ($p_{1,A}(\cdot,i)$ is a an element of $\mathfrak P([2])$ if $A$ is unitary) that
\begin{align}
\langle v_g,A^{\otimes n}P_{g'}A^{\dag\otimes n}v_g\rangle&\leq\mathrm{pl}(n)2^{nH(\bar g')}|\langle e_{x^n},A^{\dag\otimes n}v_g\rangle|^2\\
&=\mathrm{pl}(n)2^{n H(\bar g')-H(\bar g)}|\sum_{y^n\in T_g}\langle e_{x^n},A^{\dag\otimes n}e_{y^n}\rangle|^2\\
&=\mathrm{pl}(n)2^{n H(\bar g')-H(\bar g)}\cdot\left|\sum_{g_1,g_2}G^{g'\to g}_{g_1,g_2}\cdot|T_{g_1}|\cdot|T_{g_2}|\cdot\prod_{i,j=1}^2\langle e_i,A^\dag e_j\rangle^{g_i(j)}\right|^2\\
&\leq\mathrm{pl}(n)\max_{g_1+g_2=g}2^{n(H(\bar g')-H(\bar g)+2\sum_{i=1}^2\bar g'(i)H(\bar g_i)+\sum_{i,j=1}^2\bar g'(i)\bar g_i(j)\log|\langle e_i,A^\dag e_j\rangle|^2)},
\end{align}
and since an equivalent lower bound can be established as well we obtain that for two sequences $(g^{(n)})_{n\in\nn}$ and $(\hat g^{(n)})_{n\in\nn}$ with respective normalized limits $q$ and $\hat q$ we will have
\begin{align}
\lim_{n\to\infty}\langle v_{g^{(n)}},P_{\hat g^{(n)}}'v_{g^{(n)}}\rangle&=-H(q)+\\
&+\max_{W:W(\hat q)=q}\sum_{i=1}^2\hat q(i)\left(-\log\hat q(i)-\sum_{j=1}^2w(j|i)\log w(j|i)+\sum_{j=1}^2w(j|i)\log p_{1,A}(i,j)\right)\\
&=-H(q)+\max_{W:W(\hat q)=q}\sum_{i,j=1}^2\hat q(i)w(j|i)\left(-\log\hat q(i)-\log w(j|i)+\log p_{1,A}(i,j)\right)\\
&=-H(q)-\min_{r\in\Xi}D(r\|p_{1,A}),
\end{align}
where $\Xi:=\{p\in\mathfrak P([2]\times[2]):p_1=\hat q,\ p_2=q\}$. Therefore
\begin{align}
\Theta_1(A,p,q)=H(q)+\min_{r\in\Xi}D(r\|p_{1,A})
\end{align}
We have thus identified the building blocks of $\langle v,P_{f'}'v\rangle$ and can now write
\begin{align}
\Theta(p,q,s,A)=\min_{W(p)=q}\left(\hat s(1)\Theta_1(A,(\tfrac{p(1)-s(2)}{s(1)-s(2)},\tfrac{p(2)-s(2)}{s(1)-s(2)},W(\delta_1))+\hat s(2)\Theta_2(A,W(\delta_2))\right)
\end{align}
\end{proof}
\end{section}

\begin{section}{Axioms\label{sec:proof-that-axioms-are-fulfilled}}
We give a short overview over elementary properties that make the $R_t$ candidates for relative entropies. We note again that various other possibilities exist to define one-parameter families of unitary transformations - e.g. via utilization of the geodesic (see \cite{bhatia-positive}) $t\mapsto\gamma_{\rho,\sigma}(t)$ where $\gamma_{\rho,\sigma}(t):=\rho^{1/2}\left(\rho^{-1/2}\sigma\rho^{-1/2}\right)^t\rho^{1/2}$ which (upon normalization) draws a path between $\rho$ and $\sigma$ that (like the definition that we use here) enables one to uniquely define a pinching of $\rho$ to the eigenbasis of $\gamma(t)$ whenever $[\rho,\sigma]=0$ holds. Since the geodesic curve obeys $\gamma_{U\rho U^\dag,U\sigma U^\dag}(t)=U\gamma_{\rho,\sigma}(t)U^\dag$ for all states $\rho$, $\sigma$ and $t\in[0,1]$ and unitary transformations $U$ this definition leads to another quantity, call it $\tilde R_t$, which is unitarily invariant just like $R_t$ was. We leave further investigations of these connections to future work and look at some properties of the $R_t$ family:
\\\\
{\bf Continuity.} Continuity follows directly from the fact that both functions can be rewritten as a convex optimization problem where both the function to be optimized and the convex set that it is being optimized over depend continuously on $\rho$ and $\sigma$.
\\\\
{\bf Unitary invariance.} The choice of the basis $B$ in which the $P_{f,\lambda}$ are defined is just such that it changes as $B_{\alpha,\beta}(U\rho U^\dag,U\sigma U^\dag)=U B_{\alpha,\beta}(\rho,\sigma)U^\dag$. Also, the value $f$ does not change since the transformation affects both $\rho$ and the basis that it is pinched onto in the very same manner.
\\\\
{\bf Normalization.} Let $d=1$. For every $t\in[0,1]$ we have $R_t(r\|s)=-\log(s/r)$, so that $R_t(1\|1/2)=1$ and normalization is given.
\\\\
{\bf Order axiom.} Let $\rho\leq\sigma$. Then for every of the unitary matrices $U_t$ we get $U_t\rho U_t^\dag\leq U_t\sigma U_t^\dag$. Thus for every $n\in\nn$, $t\in[0,1]$ as well as every pair $(f,\lambda)$,
\begin{align}
-\tfrac{1}{n}\log\tr\{P_{f,\lambda}(\tr\{\rho^{-1}\}U_t\sigma U_t^\dag)^{\otimes n}\}&\leq-\tfrac{1}{n}\log\tr\{P_{f,\lambda}(U_t\tr\{\rho^{-1}\}\rho U_t^\dag)^{\otimes n}\}\\
&=-\tfrac{1}{n}\log\tr\{P_{f,\lambda}(U_t\bar\rho U_t^\dag)^{\otimes n}\}.
\end{align}
This implies that, for every $t\in[0,1]$,
\begin{align}
R_t(\rho\|\sigma)\leq D(\bar\rho\|\bar\rho)=0
\end{align}
and since $D(\bar\rho\|\bar\rho)=0$ we get $R_t(\rho\|\sigma)\geq0$.\\
In case that $\rho\geq\sigma$ we get $R_t(\rho\|\sigma)\geq0$ in the very same manner.
\\\\
{\bf Additivity.} Is given by definition.
\\\\
{\bf Generalized mean value axiom.} Since we consider the case $d=2$ only, we get relieved from a heavy burden: Within $\mathcal B(\mathbb C^2)$ the notion $R_t(\rho\oplus\tau\|\sigma\oplus\omega)$ implies that $\rho$, $\tau$, $\sigma$ and $\omega$ are rank-one operators which satisfy $\rho=r\cdot|u\rangle\langle u|$, $\sigma=s|u\rangle\langle|$, $\tau=t|v\rangle\langle v|$ and $\omega=w|v\rangle\langle v|$ for thwo orthogonal and normalized vectors $u,v\in\mathbb C^2$. This immediately implies that $[\rho+\tau,\sigma+\omega]=0$ and also $[\rho,\sigma]=[\tau,\omega]=0$. This implies that all the functions $R_t$ that occur are calculated as if they were classical Kullback-Leibler divergences. We thus see that the generalized mean value axiom cannot hold since it does not hold for the classical Kullback-Leibler divergence.
\\\\
{\bf Data Processing Inequality.} Proving that DPI is valid (if that is true) seems a challenging and potentially fruitful task. While it is certainly clear that $R_1=D$ satisfies the DPI, we are not yet aware of the methods which could be used to prove that DPI holds for other $t\in[0,1]$. One way to do so would certainly be to employ results from representation theory, while another obvious way would be to prove that $R_t$ equals $D_{\alpha,z}$ for certain choices of parameters. Be aware though that it has been proven in \cite{audenaert-datta} that DPI does not hold for $t=0$.
\end{section}
\begin{section}{Conclusion}
We have brought forward our approach from \cite{noetzel-hypothesis} and proven that it leads to nontrivial connections between quantum information theory, representation theory and matrix analysis. Specifically, we have:
\begin{enumerate}
\item delivered an operational interpretation for the limit $\lim_{\alpha\to1}\hat D_\alpha$ in Theorem \ref{main-theorem}
\item and defined a new and nontrivial class $\{R_t\}_{t\in[0,1]}$ of functions on qubits that are intimately connected to quantum relative entropies in Definition \ref{def:relative-entropies}. Our ability to define these functions is based on Theorem \ref{theorem:characterization-of-DU}.
\item We used our approach to guess a nontrivial formula for minors of unitary matrices in Lemma \ref{lemma:cauchy-binet-type-formula}.
\item We have additionally provided explicit formulas for certain Hilbert-Schmidt scalar products in Theorems \ref{theorem:scalar-products-for-arbitrary-d} and \ref{theorem:scalar-products-for-d=2}. These may turn out to be useful in later work.
\end{enumerate}
We had to leave open the question of a more generic connection between the class $\{R_t\}_{t\in[0,1]}$ and the $\alpha-z$ relative entropy as well as further connections to matrix geometry. It is our hope that such connections could lead to an expansion of our definition to arbitrary $d$ and that this connection would in turn be able to provide meaningful statements on the intersection between quantum information theory and representation theory.
\end{section}
\ \\\\
\emph{Acknowledgement.}
This work was supported by: the BMBF via grant 01BQ1050, the DFG via grant NO 1129/1-1.\\
The hospitality of the Isaac Newton Institute for Mathematical Sciences and stimulating discussions with Koenraad Audenaert are gratefully acknowledged.\\
Further funding was provided by the ERC Advanced Grant IRQUAT, the Spanish MINECO Project No. FIS2013-40627-P and the Generalitat de Catalunya CIRIT Project No. 2014 SGR 966.


\begin{thebibliography}{99}

\bibitem{audenaert-datta} K. Audenaert, N. Datta, ``$\alpha$-z-Relative R\'enyi Entropies'', \emph{J. Math. Phys.} 56, 022202 (2015)

\bibitem{bacon-chuang-harrow} D. Bacon, I.L. Chuang, A.W. Harrow, ``The Quantum Schur Transform: I. Efficient Qudit Circuits'', \emph{Proceedings of the eighteenth annual ACM-SIAM symposium on Discrete algorithms (SODA)}, 1235-1244 (2007)

\bibitem{beigi} S. Beigi, ``Sandwiched R\'enyi divergence satisfies data processing inequality'', \emph{J. Math. Phys.} Vol. 54, 122202 (2013)

\bibitem{bhatia-positive} R. Bhatia, \emph{Positive Definite Matrices}, Princeton University Press (2007)

\bibitem{cbn} N. Cai, H. Boche, J. N\"otzel, ``The Quantum Channel with Random State Parameters Known to the Sender'', \emph{preprint, arXiv:1506.06479} (2015)

\bibitem{christandl-thesis} M. Christandl, ``The Structure of Bipartite Quantum States - Insights from Group Theory and Cryptography'', \emph{PhD thesis, arXiv identifier: 0604183} (2006)

\bibitem{csiszar-koerner} I. Csiszar, J. K\"orner, \emph{Information Theory; Coding Theorems for Discrete Memoryless Systems}, Akad\'emiai Kiad\'o, Budapest/Academic Press Inc., New York 1981

\bibitem{csiszar-I-projections} I. Csiszar, F. Matousek,, ``Information Projections Revisited'', \emph{IEEE Trans. Inf. Theory}, Vol. 49, No. 6, 1474-1490 (2003)

\bibitem{ding-zhou} J. Ding, A. Zhou ``Eigenvalues of rank-one updated matrices with some applications'', \emph{Applied Mathematics Letters} Vol. 20, 1223-1226 (2007)

\bibitem{fulton} W. Fulton, ``Young Tableaux With Applications to Representation Theory and Geometry'', \emph{Cambridge University Press} (1997)

\bibitem{harrow-thesis} A. Harrow, ``Applications of coherent classical communication and the Schur transform to quantum information theory'' \emph{PhD-thesis, arXiv identifier: 0512255} (2005)

\bibitem{hayashi-tomamichel} M. Hayashi, M. Tomamichel, ``Correlation Detection and an Operational Interpretation of the R\'enyi Mutual Information'', \emph{preprint, arXiv identifier: 1408.6894} (2014)

\bibitem{horn-johnson} R.A. Horn, C.R. Johnson, ``Matrix analysis'', \emph{Cambridge University Press}, 4th edition, (1990)

\bibitem{kurras} S. Kurras, ``Symmetric Iterative Proportional Fitting'', \emph{Proceedings of the Eighteenth International Conference on Artificial Intelligence and Statistics} 526–534 (2015)

\bibitem{mosonyi} M. Mosonyi, ``Inequalities for the quantum R\'enyi divergences with applications to compound coding problems'', \emph{preprint, arXiv identifier: 1310.7525} (2013)

\bibitem{mosonyi-ogawa-2} M. Mosonyi, T. Ogawa, ``Strong converse exponent for classical-quantum channel coding'' \emph{preprint, arXiv identifier: 1409.3562} (2014)

\bibitem{mosonyi-ogawa} M. Mosonyi, T. Ogawa, ``Quantum hypothesis testing and the operational interpretation of the quantum R\'enyi relative entropies'', \emph{Comm. Math. Phys.} Vol. 334, Issue 3, 1617-1648 (2015)

\bibitem{lennert-dupuis-szehr-fehr-tomamichel} M. M\"uller-Lennert, F. Dupuis, O. Szehr, S. Fehr, M. Tomamichel, ``On quantum R\'eny entropies: A new definition and some properties'', \emph{J. Math. Phys.} Vol. 54, 122203 (2013)

\bibitem{noetzel-2ptypicality} J. N\"otzel, ``A solution to two party typicality using representation theory of the symmetric group'', \emph{preprint, arXiv identifier: 1209.5094} (2012)

\bibitem{noetzel-hypothesis} J. N\"otzel, ``Hypothesis testing on invariant subspaces of the symmetric group: part I. Quantum Sanov's theorem and arbitrarily varying sources'', \emph{J. Phys. A: Math. Theor.} 47 235303 (2014)

\bibitem{simon} S.M. Lin, M. Tomamichel, ``Investigating properties of a family of quantum Rényi divergences'', \emph{Quant. Inf. Proc.} Vol. 14, Issue 4, 1501-1512 (2015)

\bibitem{sternberg} S. Sternberg, \emph{Group Theory and Physics}, Cambridge University Press (1994)

\bibitem{wilde-winter-yang} M.M. Wilde, A. Winter, D. Yang, ``Strong converse for the classical capacity of entanglement-breaking and Hadamard channels'',  \emph{Comm. Math. Phys.} Vol. 331, 583-622 (2014)

\end{thebibliography}
\end{document}